\documentclass[10pt, two column, twoside]{IEEEtran}
\usepackage[utf8]{inputenc} 
\usepackage[T1]{fontenc}
\usepackage{url}
\usepackage{ifthen}
\usepackage[cmex10]{amsmath}
\usepackage{array} 

\usepackage{enumerate}
\usepackage{amssymb}
\usepackage{amsmath} 
\usepackage[lined,boxed,commentsnumbered, ruled]{algorithm2e}
\usepackage{mathrsfs}
\usepackage{bm}
\usepackage{tikz}
\usepackage{algorithmic}
\usetikzlibrary{arrows}
\usepackage{subfigure}
\usepackage{graphicx,booktabs,multirow}
\usepackage{epstopdf}
\usepackage{subfigure}
\usepackage{multirow}
\usepackage{tabularx} 
\usepackage{booktabs}
\DeclareMathOperator*{\argmin}{arg\,min}

\definecolor{colorhkust}{RGB}{20,43,140}
\definecolor{colortsinghua}{RGB}{116,52,129}
\definecolor{color1}{RGB}{128,0,0}

\usepackage{amsthm}

\newtheorem{lemma}{Lemma}
\newtheorem{theorem}{Theorem}

\newtheorem{definition}{Definition}
\newtheorem{remark}{Remark}

\newcommand{\diagg}{\mathrm{diag}}

\begin{document}

        \title{Nonconvex Demixing from Bilinear Measurements}

   \author{Jialin Dong, \textit{Student Member}, \textit{IEEE}, and Yuanming~Shi, \textit{Member}, \textit{IEEE}   \thanks{J. Dong and Y. Shi are with the School of Information Science and Technology, ShanghaiTech University, Shanghai 201210, China (e-mail: \{dongjl, shiym\}@shanghaitech.edu.cn).}

    }
        
        \maketitle

\begin{abstract}
        We consider the  problem of demixing a sequence of 
        source signals from the sum of noisy bilinear measurements. It is a generalized mathematical model for blind demixing with blind deconvolution, which is prevalent across the areas of dictionary learning, image processing, and communications. However, state-of-the-art convex methods for  blind demixing via semidefinite programming are computationally infeasible for large-scale problems. Although the existing nonconvex algorithms are able to address the scaling issue, they normally  require proper regularization  to establish optimality guarantees. The additional regularization  yields tedious
algorithmic parameters and pessimistic convergence rates with conservative step sizes.
To address the limitations of existing methods, we thus develop a provable nonconvex demixing procedure via Wirtinger flow, much like vanilla gradient descent, to harness the  benefits of {{regularization free, fast convergence rate with aggressive step size and computational optimality guarantees}}. This is achieved by exploiting the benign geometry of the blind demixing problem, thereby revealing that Wirtinger flow enforces the regularization-free iterates in the region of strong convexity and qualified level of smoothness where the step size can be chosen aggressively. 

\end{abstract}

\section{Introduction}
Demixing a sequence of source signals from the sum of bilinear measurements provides a generalized mathematical modeling framework for blind demixing with blind deconvolution \cite{ling2015blind, jung2017blind,dong}. It spans a wide scope of applications ranging from communication \cite{wang1998blind}, imaging \cite{campisi2017blind}, and machine learning \cite{bristow2013fast}, to the recent application in the context of the Internet-of-Things for sporadic and short messages communications over unknown channels \cite{dong}.  Although blind demixing can be regarded as a variant of blind deconvolution \cite{ahmed2014blind} by extending the problem of  ``single-source'' setting to the ``multi-source'' setting, it is nontrivial to accomplish the extension. The main reason is that  the ``incoherence'' between different sources brings unique challenges to develop effective algorithms for blind demixing with theoretical guarantees \cite{ling2015blind, jung2017blind, ling2018regularized}. In addition, the bilinear measurements in the blind demixing problem hamper the extension of the results for the demixing problem with linear measurements  \cite{mccoy2014sharp}.
 Moreover, the demixing procedure often involves solving highly nonconvex optimization problems which are generally dreadful to tackle. In particular,  local stationary points bring severe challenges since  it is usually intractable to even check local optimality for a feasible point \cite{ma2017implicit}.

Despite the general intractability, recent years have seen progress on convex relaxation approach for demixing problems.
 Specifically, sharp recovery bound for convex demixing with linear measurements has been established in \cite{amelunxen2014living} based on the integral geometry technique  \cite{amelunxen2014living} for analyzing the convex optimization problems with random constraints. Moreover, by lifting the original bilinear model into the the linear model with rank-one matrix, the provable convex relaxation approach for solving the blind deconvolution problem via semidefinite programming has been developed in \cite{ahmed2014blind}. Ling {\it{et al.}} in \cite{ling2015blind} further extended the theoretical analysis for blind deconvolution with single source  \cite{ahmed2014blind} to the blind demixing problem with multiple sources. The theoretical guarantees for blind demixing have been recently improved in \cite{jung2017blind}, which are built on the concept of restricted isometry property originally introduced in \cite{candes2005decoding}. Despite attractive theoretical guarantees, such convex relaxation methods fail in the high-dimensional data setting due to the high computational and storage cost for solving large-scale semidefinite programming problems.

To address the scaling issue of the convex relaxation approaches, a recent line of works has investigated  computationally efficient methods based on nonconvex optimization paradigms with  theoretical guarantees. For  high-dimensional estimation problems via nonconvex optimization methods,
state-of-the-art results can be divided into two categories, i.e., local geometry and global geometry. In the line of works that focuses on the local geometry, one shows that iterative algorithm converges to global solution rapidly when the initialization is close to the ground truth. The list of this line of successful works includes matrix completion \cite{Luo_TIT16}, phase retrieval \cite{candes2015phase,chen2015solving, ma2017implicit}, blind deconvolution \cite{li2016rapid} and blind demixing \cite{ling2018regularized}. 
The second line of works explores the global landscape of the objective function and aims to show that all local minima are globally optimal under suitable statistical conditions while the saddle points can be escaped efficiently via nonconvex iterative procedures with random initialization. The successful examples include matrix sensing \cite{bhojanapalli2016global}, matrix completion \cite{ge2017no}, dictionary learning \cite{sun2017complete}, tensor decomposition \cite{ge2015escaping}, synchronization problem \cite{bandeira2016low} and learning shallow neural networks \cite{soltanolkotabi2017theoretical}.

The  nonconvex optimization paradigm for high-dimensional estimation has also recently been applied in the setting of  blind demixing. Specifically, 
a nonconvex Riemannian optimization algorithm was developed in \cite{dong} by exploiting the manifold geometry of fixed-rank matrices. However, due to complicated iterative strategies of in the Riemannian trust-region algorithms, it is challenging to provide high-dimensional statistical analysis for such nonconvex strategy. Ling {\it{et al.}} in \cite{ling2018regularized}  developed a    regularized gradient descent procedure to optimize the nonconvex loss function directly, in which the regularization accounts for guaranteeing incoherence.
Although the regularized nonconvex  procedure in \cite{ling2018regularized} provides appealing computational properties with optimality guarantees, it usually introduces tedious algorithmic parameters that need to be carefully tuned. Moreover, theoretical analysis in \cite{ling2018regularized} provides a pessimistic convergence rate with a severely conservative step size.

In contrast, the Wirtinger flow algorithm \cite{candes2015phase},  which consists of spectral initialization and vanilla gradient descent updates without regularization, turns out to yield theoretical guarantees for important high-dimensional statistical estimation problems. In particular, the optimality guarantee for phase retrieval  was established in \cite{candes2015phase}. However, the theoretical results in \cite{candes2015phase} only ensure that the iterates
of the Wirtinger flow algorithm remain in the $\ell_2$-ball, in which the step size is chosen conservatively, yielding slow convergence rate. The statistical and computational efficiency was further improved in \cite{chen2015solving} via the truncated Wirtinger flow by carefully controlling search directions, much like regularized gradient descent. To harness all benefits of regularization free, fast convergence rates with aggressive step size and computational optimality guarantees, Ma {\it{et al.}} \cite{ma2017implicit} has recently uncovered that the Wirtinger flow algorithm (without regularization) implicitly enforces iterates within the intersection between $\ell_2$-ball and the incoherence region, i.e., the region of incoherence and contraction, for the nonconvex estimation problems of phase retrieval, low-rank matrix completion, and blind deconvolution. By exploiting the local geometry in such a region, i.e., strong convexity and qualified level of smoothness, the step size of the iterative algorithm can be chosen more aggressively, yielding faster convergence rate.  

In the present work, we extend the knowledge of implicit regularization in the nonconvex statistical estimation problems \cite{ma2017implicit} by studying the unrevealed blind demixing problem. It turns out that, for the blind demixing problem, our theory suggests a more aggressive step size for the Wirtinger flow algorithm compared with the results in \cite{ling2018regularized}, yielding substantial computational savings for blind demixing problem. The extension turns out to be nontrivial since the ``incoherence'' between multiple sources for blind demixing leads to distortion to the statistical property in the single source scenario for blind deconvolution. The similar challenge has also been observed in \cite{ling2015blind, jung2017blind} by extending the convex relaxation approach (i.e., semidefinite programming) for blind deconvolution to the setting of blind demixing. Furthermore, the noisy measurements  also  bring additional challenges
to establish theoretical guarantees. The extra technical details involved
in this paper to address these challenges shall be demonstrated clearly during the presentation. 

\subsubsection*{Notations}Throughout this paper, $f({n}) =  O(g(n))$ or $f(n)\lesssim g(n)$ denotes that there exists a constant $c>0$ such that $|f(n)|\leq c|g(n)|$ whereas $f(n)\gtrsim g(n)$ means that there exists a constant $c>0$ such that $|f(n)|\geq c|g(n)|$. $f(n)\gg  g(n)$ denotes that there exists some sufficiently large constant $c>0$ such that $|f(n)|\geq c|g(n)|$. In addition, the notation $f(n)\asymp g(n)$ means that there exist constants $c_1, c_2>0$ such that $c_1|g(n)|\leq |f(n)|\leq c_2|g(n)|$.

\section{Problem Formulation}\label{form}
In this section, we present mathematical model of the blind demixing problem in the noisy scenario. As this problem is highly intractable without any further structural assumptions, the coupled signals are thus assumed to belong to known subspaces \cite{ling2015blind, jung2017blind,ling2018regularized}. 

Let $\bm{A}^*$ denote the conjugate transpose of matrix $\bm{A}$.
Suppose we have $m$ bilinear measurements $y_j$'s, which are represented in the frequency domain as
\setlength\arraycolsep{2pt} 
\begin{align}
y_j  = \sum_{i=1}^{ s}\bm{b}_j^*\bm{h}_i^{\natural}\bm{x}_i^{\natural\ast}\bm{a}_{ij}+{e}_j,~1\leq j \leq m,
\end{align}
where $\bm{a}_{ij}\in\mathbb{
        C}^K$ and $\bm{b}_j\in\mathbb{C}^K$ are known design vectors, ${e}_j\sim\mathcal{N}(0,\frac{\sigma^2d_0^2}{2m})+i\mathcal{N}(0,\frac{\sigma^2d_0^2}{2m})$ is the additive white complex Gaussian noise with $d_0 = \sqrt{\sum_{i=1}^{s}\|\bm{h}_i^\natural\|^2_2\|\bm{x}_i^\natural\|^2_2}$ and $1/\sigma^2$ as the measurement of noise variance \cite{ling2018regularized}. Each $\bm{a}_{ij}$ is assumed to follow an i.i.d. complex Gaussian distribution, i.e., $\bm{a}_{ij}\sim \mathcal{N}(0,\frac{1}{2}\bm{I}_K)+i\mathcal{N}(0,\frac{1}{2}\bm{I}_K)$. The first $K$ columns of the unitary discrete Fourier transform (DFT) matrix $\bm{F}\in \mathbb{C}^{m\times m}$ with $\bm{FF}^* = \bm{I}_m$ form the matrix $\bm{B}:=[\bm{b}_1,\cdots,\bm{b}_m]^*\in \mathbb{C}^{m\times K}$ \cite{ling2018regularized}. Based on the above bilinear model, our goal is to simultaneously recover the underlying signals $\bm{h}_i^{\natural}\in\mathbb{C}^K$'s and $\bm{x}_i^{\natural}\in\mathbb{C}^K$'s by solving the following blind demixing problem \cite{dong, ling2018regularized}
\begin{eqnarray}
\label{lea_squ_sdp}
\mathscr{P}: \mathop{\textrm{minimize }}_{\{\bm{h}_i\},\{\bm{x}_i\}}f(\bm{h},\bm{x}):=\sum_{j=1}^m\Big|\sum_{i=1}^{ s}\bm{b}_j^*\bm{h}_i\bm{x}_i^{\ast}\bm{a}_{ij}-{y}_j\Big|^2.
\end{eqnarray}
To simplify the presentation, we denote $f(\bm{z}) := f(\bm{h},\bm{x})$, where
$
\bm{z} = \left[
\bm{z}_1^*\cdots\bm{z}_s^*
\right]^*\in \mathbb{ C}^{2sK}~{\textrm{with}}~\bm{z}_i = \left[\bm{h}_i^*~\bm{x}_i^*\right]^*\in \mathbb{ C}^{2K}.
$
We further define the discrepancy between the estimate $\bm{z}$ and the ground truth $\bm{z}^{\natural}$ as the distance function, given as
\begin{align}\label{dist}
\mbox{dist}(\bm{z},\bm{z}^{\natural})=\left(\sum_{i=1}^s\mbox{dist}^2(\bm{z}_i,\bm{z}_i^\natural)\right)^{1/2},
\end{align}
where $\mbox{dist}^2(\bm{z}_i,\bm{z}_i^\natural) = \min\limits_{\alpha_i\in\mathbb{ C}}({{\|\frac{1}{\overline{\alpha_i}}\bm{h}_i-\bm{h}_i^{\natural} \|_2^2+\|\alpha_i \bm{x}_i - \bm{x}_i^{\natural}\|_2^2 }})/{d_i}$ for $i = 1,\cdots,s$. Here, $d_i = \|\bm{h}_i^\natural\|^2+\|\bm{x}_i^\natural\|^2$ and each $\alpha_i$ is the alignment parameter. 
\section{Main Results}
In this section, we  shall present the Wirtinger flow algorithm along with the statistical analysis for blind demixing  $\mathscr{P}$.
\subsection{Wirtinger Flow Algorithm}
The Wirtinger flow algorithm  \cite{candes2015phase} is a two-stage approach consisting of spectral initialization and vanilla gradient descent update procedure without regularization. Specifically, the gradient step in the second stage of  Wirtinger flow is characterized by the notion of Wirtinger derivatives \cite{candes2015phase}, i.e., the derivatives of real valued functions over complex variables. For each $i = 1,\cdots,s$, $\nabla_{\bm{h}_i}f(\bm{h},\bm{x})$ and $\nabla_{\bm{x}_i}f(\bm{h},\bm{x})$ denote the Wirtinger gradient of $f(\bm{z})$ with respect to $\bm{h}_i$ and $\bm{x}_i$ respectively as follows:
\begin{subequations}\label{gradient}
        \begin{align}
        \nabla_{\bm{h}_i}f(\bm{z}) &= \sum_{j=1}^m\bigg(\sum_{k=1}^{ s}\bm{b}_j^*\bm{h}_k\bm{x}_k^{\ast}\bm{a}_{kj}-{y}_j\bigg)\bm{b}_j\bm{a}_{ij}^*\bm{x}_i,\label{g1}\\
        \nabla_{\bm{x}_i}f(\bm{z}) &= \sum_{j=1}^m\overline{\bigg(\sum_{k=1}^{ s}\bm{b}_j^*\bm{h}_k\bm{x}_k^{\ast}\bm{a}_{kj}-{y}_j\bigg)}\bm{a}_{ij}\bm{b}_j^*\bm{h}_i.\label{g2}
        \end{align}
\end{subequations}

\begin{algorithm}[t]\label{algo}
        \caption{Wirtinger flow  for blind demixing $\mathscr{P}$}
        \begin{algorithmic}[1]
                \renewcommand{\algorithmicrequire}{\textbf{Given:}}
                \renewcommand{\algorithmicensure}{\textbf{Output:}}
                
                \REQUIRE $\{\bm{a}_{ij}\}_{1\leq i \leq s, 1\leq j\leq m}$, $\{\bm{b}_j\}_{1\leq j \leq m},$ and $\{y_j\}_{1\leq j\leq m}
                $.
                
                \STATE \textbf{Spectral Initialization:}\renewcommand{\algorithmicdo}
                {\textbf{do in parallel}}
                \FORALL{$i = 1,\cdots,s$}       
                \STATE  Let $\sigma_1(\bm{M}_i),~\check{\bm{h}}_i^0$ and $\check{\bm{x}}_i^0$ be the leading singular value, left singular
                vector and right singular vector of matrix $\bm{M}_i:=\sum_{j=1}^m y_j\bm{b}_j\bm{a}_{ij}^*,$
                respectively.
                \STATE Set $\bm{h}_i^0 = \sqrt{\sigma_1(\bm{M}_i)} \check{\bm{h}}_i^0$ and $\bm{x}_i^0 = \sqrt{\sigma_1(\bm{M}_i)} \check{\bm{x}}_i^0$.

                \ENDFOR   
                \renewcommand{\algorithmicdo}
                {\textbf{do}}
                \FORALL{$t = 1,\cdots, T$}
                \renewcommand{\algorithmicdo}{\textbf{do in parallel}}
                \FORALL{$i = 1,\cdots,s$}
                \STATE 
                $\left[\bm{h}_i^{t+1}\atop\bm{x}_i^{t+1} \right]
                = \left[\bm{h}_i^{t}\atop\bm{x}_i^{t}\right]-\eta \left[\frac{1}{\|\bm{x}^t_i\|_2^2}\nabla_{\bm{h}_i}f(\bm{h}^t,\bm{x}^t)\atop\frac{1}{\|\bm{h}_i^t\|_2^2}\nabla_{\bm{x}_i}f(\bm{h}^t,\bm{x}^t) \right]
                $

                \ENDFOR
                \ENDFOR 
        \end{algorithmic}
\end{algorithm}
The
Wirtinger flow for the blind demixing problem is presented in Algorithm \ref{algo}, in which $T>0$ is the maximum number of iterations and the constant $\eta>0$ is the step size. 

We now provide some numerical evidence by testing the performance of the Wirtinger flow algorithm for blind demixing problem $\mathscr{P}$ (\ref{lea_squ_sdp}). We first consider the blind demixing problem in the noiseless scenario in order to clearly demonstrate the effectiveness of the Wirtinger flow algorithm. Specifically, for each $K\in\{50,100,200,400,800\}$, $s = 10$ and $m = 50K$, we generate the design vectors $\bm{a}_{ij}$'s and $\bm{b}_j$'s for each $1\leq i\leq s, 1\leq j\leq m$, according to the descriptions in Section \ref{form}. The underlying signals $\bm{h}^\natural_i,\bm{x}^\natural_i \in\mathbb{C}^K$, $1\leq i\leq s$, are generated as random vectors with unit norm. With the chosen step size  $\eta = 0.1$ in all settings, Fig. \ref{fig1_1} shows the relative error 
$
{\sum_{i = 1}^s\|\bm{h}^t_i\bm{x}^{t*}_i-\bm{h}^\natural_i\bm{x}^{\natural*}_i\|_F}/{\sum_{i = 1}^s\|\bm{h}^\natural_i\bm{x}^{\natural*}_i\|_F},
$ versus the iteration count, where $\|\cdot\|_F$ denotes the Frobenius norm.
We observe that, in the noiseless case, Wirtinger flow with constant step size enjoys extraordinary linear convergence rate which rarely changes as the problem size varies. 
\begin{figure*}[htbp]
        \centering{
                \subfigure[]{\includegraphics[width=1.3in]{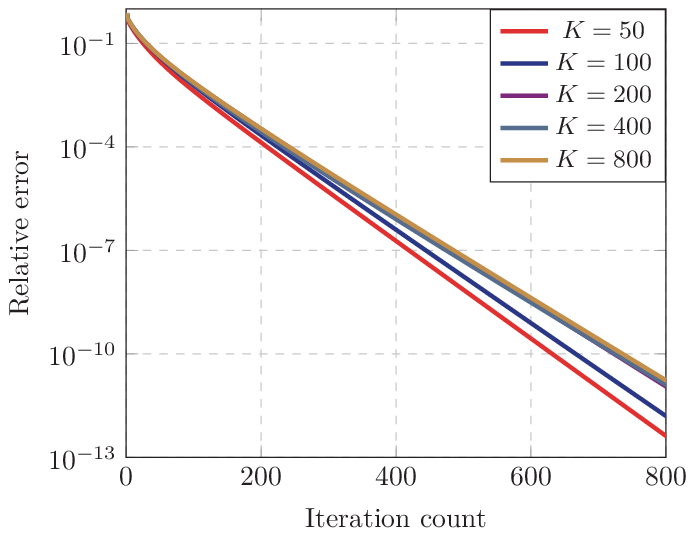}
                        \label{fig1_1}}
                \hfil
                \subfigure[]{\includegraphics[width=1.3in]{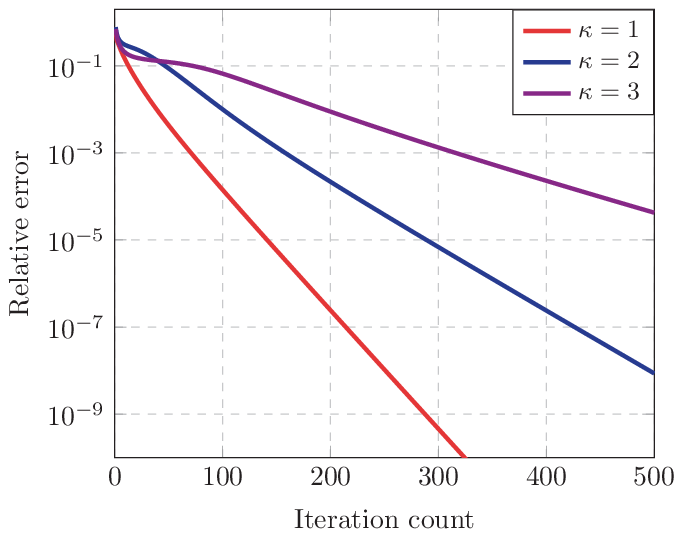}
                        \label{fig5}}
                \subfigure[]{\includegraphics[width=1.3in]{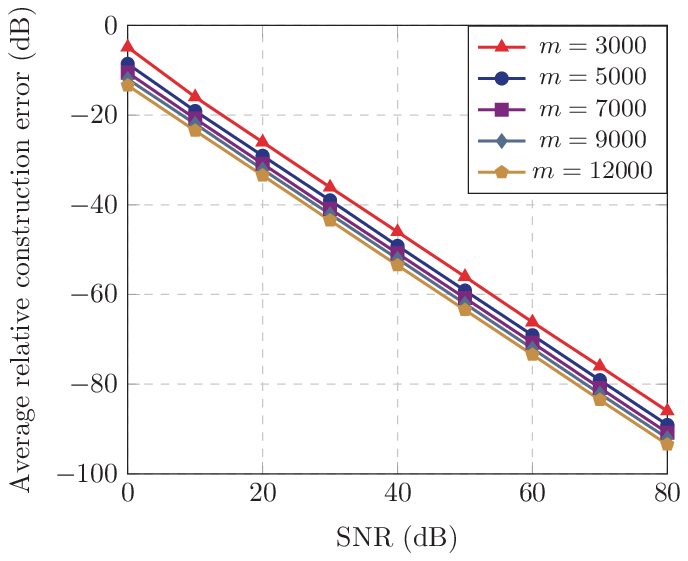}
                        \label{fig4}}      
                \subfigure[]{\includegraphics[width=1.3in]{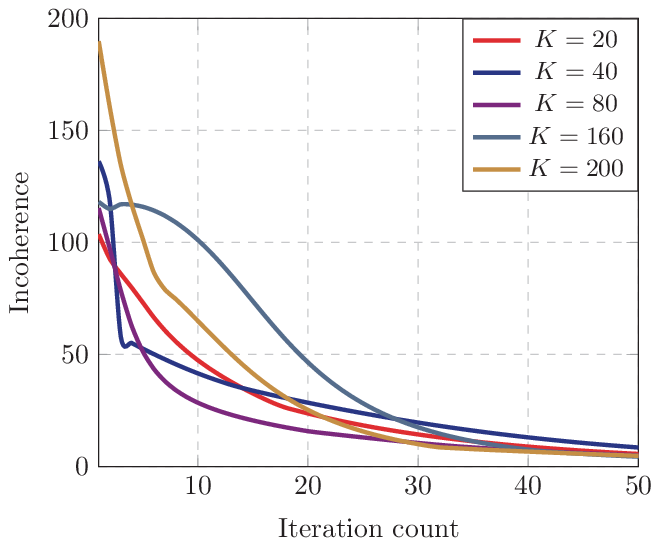}
                        \label{fig2}}
                \subfigure[]{\includegraphics[width=1.3in]{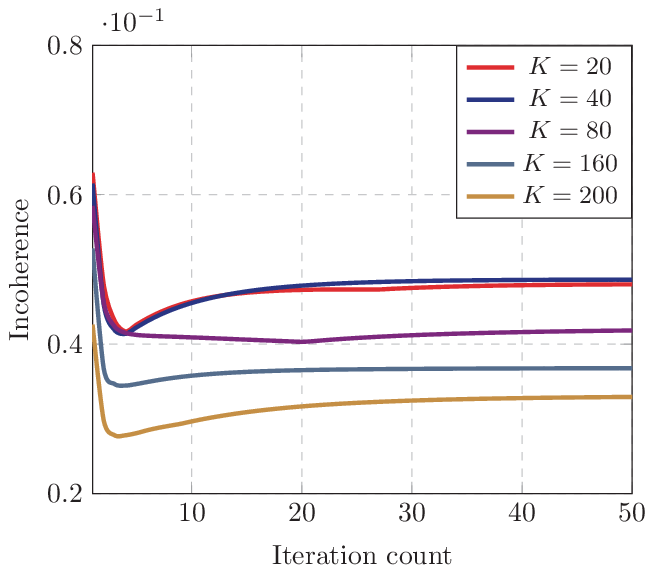}
                        \label{fig3}}

        \caption{Numerical results.}
        }       
\end{figure*}

In the noiseless scenario, we further demonstrate that the performance and convergence rate of the Wirtinger flow actually depend on the condition number, i.e., $\kappa:=\frac{\max_i \|\bm{x}_i^\natural\|_2 }{\min_i \|\bm{x}_i^\natural\|_2 }$. In this experiment, we let $K=50$, $m=800$, $s=2$, the step size be  $\eta = 0.5$ and set for the first component $\|\bm{h}_1^\natural\|_2 = \|\bm{x}_1^\natural\|_2=1$ and for the second one $\|\bm{h}_2^\natural\|_2 = \|\bm{x}_2^\natural\|_2 = \kappa$ with $\kappa\in\{1,2,3\}$. Fig. \ref{fig5} shows the relative error versus the iteration count. As we can see, the larger $\kappa$ yields slower convergence rate. This phenomenon may be caused by bad initial guess for weak components via spectral initialization \cite{ling2018regularized}. Moreover, the strong components may pollute the gradient directions for weak components, which yields slow convergence rate \cite{ling2018regularized}.
We further provide empirical results for the Wirtinger flow algorithm in the presence of noise. We set the size of source signals  $K = 50$, the sample size  $m\in\{3,5,7,9,12\}\times 10^{3}$, the user number  $s=10$, the step size  $\eta = 0.1$. The underlying signals $\bm{h}^\natural_i,\bm{x}^\natural_i \in\mathbb{C}^K$, $1\leq i\leq s$, are generated as random vectors with unit norm. Fig. \ref{fig4} shows the relative error defined above versus the signal-to-noise ratio (SNR), where the SNR is defined as
$
\text{SNR}:={\|\bm{y}\|_2}/{\|\bm{e}\|_2}
$ \cite{ling2018regularized} since it is easy to access the signal $\bm{y}$.
Both the relative error and the SNR are shown in the  dB scale. As we can see, the relative error scales linearly with the SNR, which implies that the Wirtinger flow is robust to the noise.
The main purpose of this paper is to theoretically analyze the promising empirical observations of the Wirtinger flow algorithm for  blind demixing  $\mathscr{P}$ in the  noisy scenarios. We will demonstrate that for the problem $\mathscr{P}$ the Wirtinger flow algorithm can achieve fast convergence rates with aggressive step size and computational optimality guarantees without explicit regularization.

\subsection{Theoretical Results}
Before stating the main theorem, we need to introduce the incoherence parameter \cite{ling2018regularized}, which characterizes the incoherence between $\bm{b}_j$ and $\bm{h}_i$ for $1\leq i \leq s, 1\leq j\leq m$.
\begin{definition}[Incoherence for blind demixing] Let the incoherence parameter $\mu$ be the smallest number such that
        \begin{align}\label{inco}
        \max_{1\leq i \leq s, 1\leq j\leq m}\frac{|\bm{b}^*_j\bm{h}_i^\natural|}{\|\bm{h}_i^\natural\|_2}\leq \frac{\mu}{\sqrt{m}}.
        \end{align}
\end{definition}
The incoherence  between $\bm{b}_j$ and $\bm{h}_i$ for $1\leq i \leq s, 1\leq j\leq m$ specifies the smoothness of the loss function (\ref{lea_squ_sdp}). Within the region of incoherence and contraction (defined in Section {\ref{RICss}}) that enjoys the qualified level of smoothness, the step size for iterative refinement procedure can be chosen more aggressively according to generic optimization theory \cite{ma2017implicit}. Based on the definition of incoherence,
our theory shall show that the iterates of Algorithm \ref{algo} will retain in the region of incoherence and contraction, which is endowed with strong convexity and the qualified level of smoothness.

Without loss of generality, we assume $\|\bm{h}_i^\natural\|_2 = \|\bm{x}_i^\natural\|_2$ for $i=  1,\cdots ,s$ and define the condition number $\kappa:=\frac{\max_i \|\bm{x}_i^\natural\|_2 }{\min_i \|\bm{x}_i^\natural\|_2 }\geq 1$ with $\max_i\|\bm{x}_i^\natural\|_2=1$. Define  $\mathcal{ A}_i(\bm{e}) = \sum_{j=1}^m e_j\bm{b}_j\bm{a}_{ij}^*,~i = 1,\cdots,s,
$ then the main theorem is presented in the following.

\begin{theorem}\label{mainT}
        Suppose the step size obeys $\eta>0$ and $\eta\asymp s^{-1}$, then the iterates (including the spectral initialization point) in Algorithm \ref{algo} satisfy
        \begin{subequations}
                \begin{align}
                &\mathrm{dist}(\bm{z}^t,\bm{z}^\natural)\leq C_1(1-\frac{\eta}{16\kappa})^t \Big(\frac{1}{\log^2m}-\frac{48\sqrt{s}\kappa^2}{\eta}\cdot\notag\\&\quad\max_{1\leq i \leq s}\|\mathcal{A}_i(\bm{e})\|\Big)+\frac{48C_1\sqrt{s}\kappa^2}{\eta}\max_{1\leq i \leq s}\|\mathcal{A}_i(\bm{e})\|,\label{distforT1}\\
                &\max_{1\leq i \leq s, 1\leq j\leq m}\left|\bm{a}_{ij}^*\left(\alpha_i^t\bm{x}_i^t - \bm{x}_i^{\natural}\right)\right|\cdot\|\bm{x}_i^\natural\|_2^{-1}\leq C_3\frac{1}{\sqrt{s}\log^{3/2}m},\\
                &\max_{1\leq i \leq s, 1\leq j\leq m}\left|\bm{b}_j^*\frac{1}{\overline{\alpha^t_i}} \bm{h}_i^t \right|\cdot\|\bm{h}_i^\natural\|_2^{-1}\leq C_4 \frac{\mu}{\sqrt{m}}\log^2 m,
                \end{align}
        \end{subequations}
        for all $t\geq 0$, with probability at least $1-c_1m^{-\gamma}-c_1me^{-c_2K}$ if the number of measurements
        $m\geq C(\mu^2+\sigma^2)s^2\kappa^4K\log^8 m$ for some constants $\gamma, c_1,c_2,C_1,C_3,C_4>0$ and sufficiently large constant $C>0$.
        
        Here, we denote $\alpha_i^t$ for $ i = 1,\cdots, s$ as the alignment parameter such that
        \begin{align}\label{a}
        \alpha_i^t: = \argmin_{\alpha\in\mathbb{C}}\left\| \frac{1}{\overline{\alpha}}\bm{h}_i^t-\bm{h}_i^{\natural}\right\|_2^2+ \left\|\alpha\bm{x}_i^t-\bm{x}_i^{\natural}\right\|_2^2.
        \end{align}
        In addition, with probability at least $1-O(m^{-9})$, there holds
        $
        \max_{1\leq i \leq s}\left\|\mathcal{ A}_i(\bm{e})\right\|\leq C_0\sigma \sqrt{\frac{ 10sK\log^2 m }{m}},
$
        for some absolute constant $C_0>0$ and $\sigma$ is defined in Section \ref{form}.
        
\end{theorem}
Note that the assumption of the same length of $\bm{h}_i$ and $\bm{x}_i$ only serves the purpose of simplifying the presentation. Our theoretical results can be easily extended to the scenario where $\bm{h}_i$ and $\bm{x}_i$ have different sizes. Specifically, for each $i=  1,\cdots, s,~j = 1,\cdots, m$, if $\bm{h}_i,\bm{b}_j \in\mathbb{C}^K$ and $\bm{x}_i,\bm{a}_{ij}\in\mathbb{C}^N$, the requirement of sample size turns out to be $m\geq C(\mu^2+\sigma^2)s^2\kappa^4\mathrm{max}\{K,N\}\log^8 m$.

Theorem \ref{mainT} endorses the empirical results shown in Fig. \ref{fig1_1}, Fig. \ref{fig5} and Fig. \ref{fig4}. Specifically, compared to the step size (i.e., $\eta\lesssim\frac{1}{s\kappa m}$) suggested in \cite{ling2018regularized} for regularized gradient descent, our theory yields a more aggressive step size (i.e., $\eta\asymp s^{-1}$) even without regularization. According to (\ref{distforT1}), in the noiseless scenario, the Wirtinger flow algorithm can achieve $\epsilon$-accuracy within $s\kappa\log(1/\epsilon)$ iterations, while previous theory in \cite{ling2018regularized} suggests  $ s\kappa m\log(1/\epsilon)$ iterations. In the noisy scenario, the convergence rate of the Wirtinger flow algorithm is independent of the number of measurements $m$ and related to the level of the noise.
The sample complexity, i.e., $m\geq C s^2K\mathrm{poly}\log m$ with sufficiently large constant $C>0$, is comparable to the result in \cite{ling2018regularized} which uses explicit regularization.
However, we expect to reduce the sample complexity to $m\geq C sK\mathrm{poly}\log m,$ with sufficiently large constant $C>0$ by a tighter analysis, e.g., eluding controlling terms involved $s^2/m$, which is left for future work.

  For further illustrations, we plot the incoherence measure $\max_{1\leq
        i \leq s,1\leq j \leq m}|\bm{a}_{ij}^*(\alpha_i^t{\bm{x}}_i^t-\bm{x}_i^\natural)|$ (in Fig. \ref{fig2}) and $\max_{1\leq i \leq s,1\leq j \leq m}|\bm{b}_{j}^*\frac{1}{\overline{\alpha_i^t}}{\bm{h}}_i^t|$ (in Fig. \ref{fig3}) of the gradient iterates versus iteration count, under the setting $K\in\{20,40,80,160,200\}$, $m = 50K$, $s=10$, $\eta = 0.1$, $\sigma = 10^{-1}$ with $\|\bm{h}_i^\natural\|_2 =\|\bm{x}_i^\natural\|_2 =1$ for $1\leq i\leq s$. We observe that both incoherence measures remain bounded by befitting values for all iterations.
\section{Trajectory Analysis for Blind Demixing}  
In this section, we prove the main theorem via trajectory analysis for blind demixing via the Wirtinger flow algorithm. We shall reveal that iterates of Wirtinger flow, i.e., Algorithm \ref{algo}, stay in the region of incoherence and contraction by exploiting the local geometry of blind demixing $\mathscr{P}$. The steps of proving Theorem 1 are summarized as follows.
\begin{itemize}
	\item  \textbf{Characterizing local geometry in the region of incoherence and contraction (RIC).} We first characterize a region $\mathcal{ R}$, i.e., RIC, where the objective function enjoys restricted strong convexity and smoothness near the ground truth $\bm{z}^\natural$. Moreover, any point $\bm{z}\in\mathcal{ R}$ satisfies the $\ell_2$ error contraction and the incoherence conditions. This will be established in Lemma \ref{L1}. Provided that all the iterates of Algorithm \ref{algo} are in the region $\mathcal{ R}$, the convergence rate of the algorithm can be further established, according to Lemma \ref{2}.
	\item \textbf{Constructing the auxiliary sequences via the leave-one-out approach.} To justify that the Wirtinger Flow algorithm enforces the iterates to stay within the RIC, we introduce the leave-one-out sequences. Specifically, the leave-one-out sequences are denoted by $\{\bm{h}_i^{t,(l)},\bm{x}_i^{t,(l)}\}_{t\geq 0}$ for each $1\leq i \leq s$, $1\leq l \leq m$ obtained by removing the $l$-th measurement from the objective function $f(\bm{h},\bm{x})$. Hence, $\{\bm{h}_i^{t,(l)}\}$ and $\{\bm{x}_i^{t,(l)}\}$ are independent with $\{\bm{b}_j\}$ and $\{\bm{a}_{ij}\}$, respectively.
	\item \textbf{Establishing the incoherence condition via induction}. In this step, we employ the auxiliary sequences to establish the incoherence condition via induction. That is, as long as the current iterate stays within the RIC, the next iterate remains in the RIC.
	\begin{itemize}
		\item \textbf{Concentration between original and auxiliary sequences.} The gap between $\{\bm{z}^t\}$ and $\{\bm{z}^{t,(l)}\}$ is established in Lemma \ref{3} via employing the restricted strong convexity of the objective function in RIC. 
		\item \textbf{Incoherence condition of auxiliary sequences.}Based on the fact that $\{\bm{z}^t\}$ and $\{\bm{z}^{t,(l)}\}$ are sufficiently close, we can instead bound the incoherence of $\bm{h}_i^{t,(l)}$ (resp. $\bm{x}_i^{t,(l)}$) with respect to $\{\bm{b}_{j}\}$ (resp. $\{\bm{a}_{ij}\}$), which turns out to be much easier due to the statistical independence between   $\{\bm{h}_i^{t,(l)}\}$ (resp. $\{\bm{x}_i^{t,(l)}\}$) and $\{\bm{b}_j\}$ (resp.$\{\bm{a}_{ij}\}$).
		\item \textbf{Establishing iterates in RIC.} By combining the above bounds together, we arrive at
		$|\bm{a}_{ij}^*(\bm{x}_i^{t}-\bm{x}_i^\natural)|\leq \|\bm{a}_{ij}\|_2\cdot\|\bm{x}_i^{t} - \bm{x}_i^{t,(l)}\|_2+\|\bm{a}_{ij}^*(\bm{x}_i^{t,(l)}-\bm{x}_i^\natural)\|$
		via the triangle inequality. Based on the similar arguments, the other incoherence condition will be established in Lemma \ref{4}.
		\item \textbf{Establishing initial point in RIC.} Lemma \ref{L6}, Lemma \ref{L7} and Lemma \ref{L8} are integrated to justify that the spectral initialization point is in RIC.
	\end{itemize}
	
\end{itemize}

\subsection{Characterizing Local Geometry in the Region of Incoherence and Contraction }\label{RICss}

We first introduce the notation of Wirtinger Hessian.    Specifically,   let $\overline{\bm{A}}$ denote the entry-wise conjugate of matrix $\bm{A}$ and $f_{\text{clean}}$ denote the objective function of noiseless case. 
The Wirtinger Hessian of $f_{\text{clean}}(\bm{z})$ with respect to $\bm{z}_i$ can be written as
\begin{align}\label{chess}
\nabla^2_{\bm{z}_i}f_{\text{clean}}:=\left[~{\bm{C} \atop\bm{E}^*} ~{\bm{E} \atop\overline{\bm{C}}}~\right],
\end{align} 
where
$
\bm{C}: = \frac{\partial }{\partial \bm{z}_i}\bigg( \frac{\partial f_{\text{clean}} }{\partial \bm{z}_i}\bigg)^*$ and $\bm{E}: = \frac{\partial }{\partial \bm{\overline{z}}_i}\bigg( \frac{\partial f_{\text{clean} }}{\partial \bm{z}_i}\bigg)^*.
$
The Wirtinger Hessian of $f_{\text{clean}}(\bm{z})$ with respect to $\bm{z}$ is thus represented as
$
\nabla^2f_{\text{clean}}(\bm{z}):=\diagg(\{\nabla^2_{\bm{z}_i}f_{\text{clean}}\}_{i=1}^s),
$
where the operation $\diagg(\{\bm{A}_i\}_{i=1}^s)$ generates a block diagonal matrix with the diagonal elements as the matrices $\bm{A}_1,\cdots,\bm{A}_s$.
Please refer to Appendix \ref{appda} for more details on the Wirtinger Hessian. In addition, we say
$(\bm{h}_i,\bm{x}_i)$ is aligned with $(\bm{h}_i^{\prime},\bm{x}_i^{\prime})$, if the following condition is satisfied
\begin{align}
\left\|\bm{h}_i-\bm{h}_i^{\prime}\right\|_2^2 &+ \left\|\bm{x}_i-\bm{x}_i^{\prime}\right\|_2^2=\notag\\ &
\min_{\alpha\in\mathbb{ C}}\left\{\left\|\frac{1}{\overline{\alpha}}\bm{h}_i-\bm{h}_i^{\prime}\right\|_2^2+ \left\|\alpha\bm{x}_i-\bm{x}_i^{\prime}\right\|_2^2\right\}.
\end{align}

Let $\|\bm{A}\|$ denote the spectral norm of matrix $\bm{A}$. We have the following lemma.

\begin{lemma}\label{L1}
        (Restricted strong convexity and smoothness for blind demixing problem $\mathscr{P}$). Let $\delta>0$ be a sufficiently small constant. If the number of measurements satisfies $m\gg \mu^2s^2\kappa^2K\log^5m $, then with probability at least $1-O(m^{-10})$, the Wirtinger Hessian $\nabla^2f_{\mathrm{clean}}(\bm{z})$ obeys
        \begin{align}
        \bm{u}^*\left[\bm{D}\nabla^2f_{\mathrm{clean}}(\bm{z})  + \nabla^2f_{\mathrm{clean}}(\bm{z})\bm{D}\right]\bm{u}&\geq \frac{1}{4\kappa}\|\bm{u}\|_2^2~~{\text{and}}\notag\\
        ~~\left\|\nabla^2f_{\mathrm{clean}}(\bm{z})\right\|&\leq 2+s
        \end{align}
        simultaneously for all 
        \[\begin{split}
        \bm{u} = \left[~\begin{matrix}
        \bm{u}_1\\\vdots\\\bm{u}_s
        \end{matrix}~\right]~\textrm{with}~\bm{u}_i = \left[~\begin{matrix}
        \bm{h}_i-\bm{h}_i^\prime\\
        \bm{x}_i-\bm{x}_i^\prime\\
        \overline{\bm{h}_i-\bm{h}_i^\prime}\\
        \overline{\bm{x}_i-\bm{x}_i^\prime}
        \end{matrix}~\right],
        \end{split}\]
        \[\begin{split}
        & {\text{and}}   ~\bm{D} = \diagg\left(\{\bm{W}_i\}_{i=1}^s\right)~\\
        & \mbox{with}~\bm{W}_i = \diagg\left(\left[
        \overline{\beta}_{i1}\bm{I}_K~
        \overline{\beta}_{i2}\bm{I}_K~
        \overline{\beta}_{i1}\bm{I}_K~
        \overline{\beta}_{i2}\bm{I}_K
        \right]^*     \right).
        \end{split}\]
        Here $\bm{z}$ satisfies
        \begin{subequations}\label{zcondition}
                \begin{align}
                &\max_{1\leq i \leq s}\max\left\{\|\bm{h}_i-\bm{h}_i^\natural\|_2,\|\bm{x}_i-\bm{x}_i^\natural\|_2\right\}\leq \frac{\delta}{\kappa\sqrt{s}},\label{con1}\\
                &\max_{1\leq i \leq s,1\leq j \leq m}\left|\bm{a}_{ij}^*\left(\bm{x}_i-\bm{x}_i^\natural\right)\right|\cdot\|\bm{x}_i^\natural\|_2^{-1}\leq \frac{2C_3}{\sqrt{s}\log^{3/2}m},\label{con2}\\
                &\max_{1\leq i \leq s,1\leq j \leq m}|\bm{b}_{j}^*\bm{h}_i|\cdot\|\bm{h}_i^\natural\|_2^{-1}\leq \frac{2C_4\mu}{\sqrt{m}}\log^{2}m,\label{con3}
                \end{align}
        \end{subequations}
        where $(\bm{h}_i,\bm{x}_i)$ is aligned with $(\bm{h}_i^\prime,\bm{x}_i^\prime)$, and one has
      $
        \max \{\|\bm{h}_i-\bm{h}_i^\natural\|_2,\|\bm{h}_i^\prime-\bm{h}_i^\natural\|_2,\|\bm{x}_i-\bm{x}_i^\natural\|_2,
        \|\bm{x}_i^\prime-\bm{x}_i^\natural\|_2\} \leq{\delta}/({\kappa\sqrt{s}}),
     $
        for $i=  1,\cdots,s$ and $\bm{W}_i$'s satisfy that for $\beta_{i1},\beta_{i2}\in\mathbb{ R}$, for $i=  1,\cdots,s$
$
        \max_{1\leq i \leq s}\max\left\{|\beta_{i1}-\frac{1}{\kappa}|,|\beta_{i2}-\frac{1}{\kappa}|\right\}\leq \frac{\delta}{\kappa\sqrt{s}}.
$
        Therein, $C_3,C_4\geq 0$ are numerical constants.
\end{lemma}

\begin{proof}
  Please refer to Appendix \ref{proof_L1} for details.   
\end{proof}
Conditions (\ref{con1})-(\ref{con3}) identify the local geometry of  blind demixing in the noiseless scenario. Specifically, (\ref{con1}) identifies a neighborhood that is close to the ground truth in $\ell_2$-norm. In addition, (\ref{con2}) and (\ref{con3}) specify the incoherence region with respect to the vectors $\bm{a}_{ij}$ and $\bm{b}_{j}$ for $1\leq i\leq s, 1\leq j\leq m$, respectively. This lemma paves the way to the proof of Lemma \ref{2} and Lemma \ref{3}. Specifically, the quantities of interest in these lemmas are decomposed into the part with respect to $f_{\text{clean}}$ and the part with respect to the noise $\bm{e}$ such that Lemma \ref{L1.1} can be exploited to bound the first part.

Based on the local geometry in the region of incoherence and contraction, we further establish contraction of the error measured by the distance function (\ref{dist}). 
\begin{lemma}\label{2}
        Suppose the number of measurements satisfies $m\gg\mu^2s^2\kappa^2K\log^5m$ and the step size obeys $\eta>0$ and $\eta\asymp s^{-1}$. Then with probability at least $1-O(m^{-10})$,
        $\mathrm{dist}(\bm{z}^{t+1},\bm{z}^\natural)\leq (1-\eta/(16\kappa))\mathrm{dist}(\bm{z}^t,\bm{z}^\natural)+3\kappa\sqrt{s}\max_{1\leq k \leq s}\left\|\mathcal{A}_k(\bm{e})\right\|,$
        provided that
        \begin{subequations}\label{con85}
                \begin{align}
                &\quad\quad\mathrm{dist}(\bm{z}^{t},\bm{z}^\natural)\leq \xi,\label{16a}\\
        &       \max_{1\leq i \leq s,1\leq j \leq m}\left|\bm{a}_{ij}^*\left(\widetilde{\bm{x}}_i^t-\bm{x}_i^\natural\right)\right|\cdot\|\bm{x}_i^\natural\|_2^{-1}\leq \frac{2C_3}{\sqrt{s}\log^{3/2}m},\\
                &\max_{1\leq i \leq s,1\leq j \leq m}\left|\bm{b}_{j}^*\widetilde{\bm{h}}_i^t\right|\cdot\|\bm{h}_i^\natural\|_2^{-1}\leq \frac{2C_4\mu}{\sqrt{m}}\log^{2}m,\label{con85c}
                \end{align}
                for some constants $C_3,C_4>0$ and a sufficiently small constant $\xi>0$. Here, $\widetilde{\bm{h}}^t_i$ and $\widetilde{\bm{x}} ^t_i
                $ are defined as $\widetilde{\bm{h}}^t_i = \frac{1}{\overline{\alpha_i^t}}\bm{h}_i^t$ and $\widetilde{\bm{x}} ^t_i= \alpha_i^t\bm{x}_i^t$ for $i = 1,\cdots, s$.
        \end{subequations}
\end{lemma}
\begin{proof}
        Please refer to Appendix \ref{proof2} for details.
\end{proof}
\begin{remark}\label{R2}
        The key idea of proving Lemma \ref{2} is to decompose the gradient (\ref{gradient}) to the the part of pure gradient $\nabla_{\bm{h}_i}f_{\mathrm{clean}}(\bm{z}) $ (resp. $\nabla_{\bm{x}_i}f_{\mathrm{clean}}(\bm{z}) $) and the part relative to the noise, i.e., $\mathcal{A}_i(\bm{e})\bm{x}_i$ (resp. $\mathcal{A}_i^*(\bm{e})\bm{h}_i$). The pure gradient $\nabla_{\bm{h}_i}f_{\mathrm{clean}}(\bm{z}) $ (resp. $\nabla_{\bm{x}_i}f_{\mathrm{clean}}(\bm{z}) $) is required in Lemma \ref{L1}.
\end{remark}
As a result, if $\bm{z}^t$ satisfies condition (\ref{con85}) for all $0\leq t\leq T = O(m^\gamma)$ for some arbitrary constant $\gamma>0$, then there is 
\begin{align}
&\mbox{dist}(\bm{z}^t,\bm{z}^\natural)-48\kappa^2\sqrt{s}/\eta\max_{1\leq k \leq s}\left\|\mathcal{A}_k(\bm{e})\right\|\notag\\
\leq& \rho^t(\mathrm{dist}(\bm{z}^0,\bm{z}^\natural)-48\kappa^2\sqrt{s}/\eta\max_{1\leq k \leq s}\left\|\mathcal{A}_k(\bm{e})\right\|),
\end{align}
with probability at least $1-O(m^{-\gamma})$ for some arbitrary constant $\gamma>0$, where $\rho:=1-\eta/(16\kappa)$.
In the absence of noise ($\bm{e} = \bm{0}$), exact recovery can be established and it yields linear convergence rate due to $\mbox{dist}(\bm{z}^t,\bm{z}^\natural)\leq \rho^t\mathrm{dist}(\bm{z}^0,\bm{z}^\natural)$. In addition, stable recovery can be achieved in the presence of noise, where the estimation error is controlled by the noise level.

\subsection{Establishing Iterates in the Region of Incoherence and Contraction}
In this subsection, we will demonstrate that the iterates of Wirtinger flow algorithm stay within the region of incoherence and contraction. In particular, the leave-one-out argument has been introduced to address the statistical dependence between $\{\bm{h}_i^{t}\}$ (resp. $\{\bm{x}_i^{t}\}$) and $\{\bm{b}_j\}$ (resp.$\{\bm{a}_{ij}\}$). Recall that $\{\bm{h}_i^{t,(l)},\bm{x}_i^{t,(l)}\}$ are defined in the recipe for proving Theorem 1.
For simplicity, we denote $\bm{z}^{t,(l)} =[\bm{z}_1^{t,(l)\ast} \cdots \bm{z}_s^{t,(l)\ast}]^\ast$     where $\bm{z}_i^{t,(l)} = [\bm{h}_i^{t,(l)\ast} ~ \bm{x}_i^{t,(l)\ast}
]^\ast$ and $f\left(\bm{z}^{t,(l)}\right):=f^{(l)}\left(\bm{h},\bm{x}\right)$.
We further define the alignment parameters $\alpha_i^{t,(l)}$, signals $\widetilde{\bm{h}} ^{t,(l)}_i$ and $\widetilde{\bm{x}} ^{t,(l)}_i$ in the context of leave-one-out sequence.

We continue the proof by induction. For brief, with $\widetilde{\bm{z}}_i^t =  [\widetilde{\bm{z}}_1^{t*},\cdots, \widetilde{\bm{z}}_s^{t*}]^*$ where $\widetilde{\bm{z}}_i^t =[\widetilde{\bm{h}}_i^{t\ast} ~ \widetilde{\bm{x}}_i^{t\ast}
]^\ast $, the set of induction hypotheses of local geometry is listed as follows:
\begin{subequations}\label{h}
        \begin{align}
        &\mathrm{dist}(\bm{z}^t,\bm{z}^\natural)\leq C_1\frac{1}{\log^2m}\label{h0},\\
        &\mathrm{dist}(\bm{z}^{t,(l)},\widetilde{\bm{z}}^t)\leq C_2\frac{s\kappa\mu}{\sqrt{m}}\sqrt{\frac{\mu^2K\log^9m}{m}},\label{h1}\\
        &\max_{1\leq i \leq s, 1\leq j\leq m}\left|\bm{a}_{ij}^*\left(\widetilde{\bm{x}}_i^t - \bm{x}_i^{\natural}\right)\right|\cdot\|\bm{x}_i\|_2^{-1}\leq C_3\frac{1}{\sqrt{s}\log^{3/2}m},\label{h2}\\
        &\max_{1\leq i \leq s, 1\leq j\leq m}\left|\bm{b}_l^* \widetilde{\bm{h}}_i^t \right|\cdot\|\bm{h}_i\|_2^{-1}\leq C_4 \frac{\mu}{\sqrt{m}}\log^2 m,\label{h3}
        \end{align}
\end{subequations}
where $C_1,C_3$ are some sufficiently small constants, while $C_2,C_4$ are some sufficiently large constants. In particular, (\ref{h0}) and (\ref{h1}) can be also represented with respect to $\bm{z}_i$:
\begin{subequations}
        \begin{align}
        &\mathrm{dist}(\bm{z}_i^t,\bm{z}_i^\natural)\leq C_1\frac{1}{\sqrt{s}\log^2m}\label{h01},\\
        &\mathrm{dist}(\bm{z}_i^{t,(l)},\widetilde{\bm{z}}_i^t)\leq C_2\frac{\kappa\mu}{\sqrt{m}}\sqrt{\frac{s\mu^2K\log^9m}{m}},\label{h11}
        \end{align}
\end{subequations} 
for $i= 1,\cdots, s$. We aim to specify that the induction hypotheses (\ref{h}) hold for $(t+1)$-th iteration with high probability, if these hypotheses hold up to the $t$-th iteration. 
Since (\ref{h0}) has been identified in (\ref{16a}) as $\delta\asymp 1/\log^2 m$, we begin with the hypothesis (\ref{h1}) in the following lemma.
\begin{lemma}\label{3}
        Suppose the number of measurements satisfies $m\gg (\mu^2+\sigma^2)s^{2}\kappa K\log^{13/2}m$ and the step size obeys $\eta>0$ and $\eta\asymp s^{-1}$. Under the hypotheses (\ref{h}) for the $t$-th iteration, one has
$
        \mathrm{dist}(\bm{z}^{t+1,(l)},\widetilde{\bm{z}}^{t+1})\leq C_2\frac{s\kappa\mu}{\sqrt{m}}\sqrt{\frac{\mu^2K\log^9m}{m}}$, $
        \max_{1\leq l\leq m}\left\|\widetilde{\bm{z}}^{t+1,(l)},\widetilde{\bm{z}}^{t+1}\right\|_2\lesssim C_2\frac{s\mu}{\sqrt{m}}\sqrt{\frac{\mu^2K\log^9m}{m}},
$
        with probability at least $1-O(m^{-9})$.
\end{lemma}
\begin{proof}
        Please refer to Appendix \ref{proof4} for details.
\end{proof}    
\begin{remark}
        The key idea of proving Lemma \ref{3} is similar to the one in Lemma \ref{2} that decomposes the gradient (\ref{gradient}) in the update rule into the part of pure gradient and the part relative to the noise. Combining Lemma \ref{L1} and Lemma \ref{L16}, we  finish the proof.
\end{remark}
Before proceeding to the hypothesis (\ref{h2}), let us first show the incoherence of the leave-one-out iterate $\bm{x}_i^{t+1,(l)}$ with respect to $\bm{a}_{il}$ for all $1\leq i \leq s$, $1\leq l\leq m$.  Based on the triangle inequality, one has
\begin{align}\label{33}
\|\widetilde{\bm{x}}_i^{t+1,(l)} - \bm{x}_i^{\natural}\|_2&\leq \|\widetilde{\bm{x}}_i^{t+1,(l)}-\widetilde{\bm{x}}_i^{t+1}\|+\|\widetilde{\bm{x}}_i^{t+1} - \bm{x}_i^{\natural}\|_2\notag\\
&\overset{(\text{i})}{\leq}C\frac{\mu}{m}\sqrt{\frac{\mu^2sK\log^9 m}{m}} + C_1\frac{1}{\kappa\sqrt{s}\log^2 m}\notag\\
&\overset{(\text{ii})}{\leq}{2C_1}/({\kappa\sqrt{s}\log^2m}),
\end{align}
where (i) arises from Lemma \ref{2} and Lemma \ref{3} and (ii) holds as long as $m\gg (\mu^2+\sigma^2)\sqrt{sK}\kappa^{2/3}\log^{13/2}m$.
Using the  inequality (\ref{33}), the standard Gaussian concentration inequality in \cite{ma2017implicit} and the statistical independence, it follows that 
\begin{align}\label{abov}
&\max_{1\leq i \leq s, 1\leq l\leq m}\left|\bm{a}_{il}^*\left(\widetilde{\bm{x}}_i^{t+1,(l)} - \bm{x}_i^{\natural}\right)\right|\cdot\|\bm{x}_i^{\natural}\|_2^{-1}
\notag\\
\leq& 5\sqrt{\log m}\max_{1\leq i \leq s, 1\leq l\leq m}\left\|\widetilde{\bm{x}}_i^{t+1,(l)} - \bm{x}_i^{\natural}\right\|_2\cdot\|\bm{x}_i^{\natural}\|_2^{-1}\notag\\
\leq & 10C_1 \frac{1}{\sqrt{s}\log^{3/2}m}
\end{align}
with probability exceeding $1-O(m^{-9})$.
For each $1\leq i \leq s, 1\leq l\leq m$, we further obtain
\begin{align}
&\left|\bm{a}_{il}^*\left(\widetilde{\bm{x}}_i^{t+1} - \bm{x}_i^{\natural}\right)\right|\cdot\|\bm{x}_i^{\natural}\|_2^{-1}\notag\\
\overset{(\text{i})}{\leq}&\left(\|\bm{a}_{il}\|_2\|\widetilde{\bm{x}}_i^{t+1} - \widetilde{\bm{x}}_i^{t+1,(l)}\|_2+\left|\bm{a}_{il}^*\left(\widetilde{\bm{x}}_i^{t+1,(l)} - \bm{x}_i^{\natural}\right)\right|\right)\|\bm{x}_i^{\natural}\|_2^{-1}\notag\\
\overset{(\text{ii})}{\leq}&3\sqrt{K}\cdot C\frac{\kappa\mu}{m}\sqrt{\frac{\mu^2sK\log^9 m}{m}} + 10C_1\frac{1}{\sqrt{s}\log^{3/2} m}\notag\\
\overset{(\text{iii})}{\leq}&C_3\frac{1}{\sqrt{s}\log^{3/2}m},
\end{align} 
where  step (i) is based on the Cauchy-Schwarz inequality, step (ii) follows from the bound (\ref{abov}), Lemma \ref{3} and the bound  with probability at least $1-Cm\exp(-cK)$, for some constants $c,C>0$ \cite{ma2017implicit},
$
\max_{1\leq j\leq m}\|\bm{a}_j\|_2\leq 3\sqrt{K},
$
and the last step (iii) holds as long as $m\gg (\mu^2+\sigma^2)s\kappa^{2/3} K\log^6m$ and $C_3\geq 11C_1$.
It remains to justify the incoherence of $\bm{h}_i^{t+1}$ with respect to $\bm{b}_l$ for all $1\leq i \leq s, 1\leq l \leq m$. The result is summarized as follows.
\begin{lemma}\label{4}
        Suppose the induction hypotheses (\ref{h}) hold true for $t$-th iteration and the number of measurements obeys $m\gg (\mu^2+\sigma^2)s^2K\log^8m$. Then with probability at least $1-O(m^{-9})$,
$
        \max_{1\leq i \leq s, 1\leq j\leq m}\left|\bm{b}_l^* \widetilde{\bm{h}}_i^{t+1} \right|\cdot\|\bm{h}_i^{\natural}\|_2^{-1}\leq C_4 \frac{\mu}{\sqrt{m}}\log^2 m,
$
        provided that $C_4$ is sufficiently large and the step size obeys $\eta>0$ and $\eta\asymp s^{-1}$.
\end{lemma}
\begin{proof}
        Please refer to Appendix \ref{proof5} for details.
\end{proof}
\begin{remark}
        Based on the claim (\ref{L526}) in Lemma \ref{L16}, it suffices to control $|\bm{b}_l^*\frac{1}{\overline{\alpha_i^t}}\bm{h}_i^{t+1}|\cdot\|\bm{h}_i^\natural\|_2$ in order to bound $\left|\bm{b}_l^* \widetilde{\bm{h}}_i^{t+1} \right|\cdot\|\bm{h}_i^{\natural}\|_2^{-1}$ in  Lemma \ref{4}. We represent $\frac{1}{\overline{\alpha_i^t}}\bm{h}_i^{t+1}$ by the gradient update rule where the gradient is decomposed as Remark \ref{R2} describes. The quantities of interest are  separated into several terms which are bounded individually.
        In addition, the random vector $\bm{a}_{ij}$ with i.i.d. plays a vital role in the proof since $\mathbb{E}(\bm{a}_{ij}\bm{a}_{kj}^*)=0$ for $k\neq i$.
\end{remark}

\subsection{Establishing Initial Point in the Region of Incoherence and Contraction}
In order to finish the induction step, we  need to further show that the spectral initializations $\bm{z}_i^0$ and $\bm{z}_i^{0,(l)}$ for $1\leq i\leq s, 1\leq l \leq m$ hold for the induction hypotheses (\ref{h}) of local geometry. The related lemmas are summarized as follows.
\begin{lemma}\label{L6}
        With probability at least $1-O(m^{-9})$, there exists some constant $C>0$ such that 
        \begin{align}
        \min_{\alpha_i\in\mathbb{C},|\alpha_i| = 1}\left\{\left\|\alpha_i\bm{h}_i^0-\bm{h}_i^\natural\right\|+\left\|\alpha_i\bm{x}_i^0-\bm{x}_i^\natural\right\|\right\}\leq \frac{\xi}{\kappa\sqrt{s}} ~\text{and}\label{26}\\
        \min_{\alpha_i\in\mathbb{C},|\alpha_i| = 1}\left\{\left\|\alpha_i\bm{h}_i^{0,(l)}-\bm{h}_i^\natural\right\|+\left\|\alpha_i\bm{x}_i^{0,(l)}-\bm{x}_i^\natural\right\|\right\}\leq \frac{\xi}{\kappa\sqrt{s}},\label{27}
        \end{align} 
        and $||\alpha_i^0|-1|< 1/4$, for each $1\leq i\leq s,1\leq l\leq m$, provided that
        $m\geq {C(\mu^2+\sigma^2)s\kappa^2K\log m}/{\xi^2} .$
\end{lemma}
\begin{proof}
       Please refer to Appendix \ref{proof6} for details.
\end{proof}
\begin{remark}
        The proof of Lemma 8 is based on the Wedin's sin$\Theta$ theorem \cite{dopico2000note} and the bound in \cite{ling2018regularized}, i.e., 
        for any $\xi>0$, 
        $
        \|\bm{M}_i-\mathbb{ E}[\bm{M}_i]\|\leq {\xi}/({\kappa\sqrt{s}}),
        $
        with probability at least $1-O(m^{-9})$, provided that 
        $
        m \gg {c_2(\mu^2+\sigma^2)s\kappa^2K\log m}/{\xi^2},
        $
        for some constant $c_2>0$. 
\end{remark}
From the definition of distance function (\ref{dist}) and the assumption $\xi\asymp 1/\log^2m$, we immediately imply that 
\begin{align}
&\mbox{dist}(\bm{z}^0,\bm{z}^\natural)\notag\\
\overset{(\text{i})}{\leq}&\min_{\alpha_i\in\mathbb{ C}}\sqrt{s}\kappa\left\{\left\|\frac{1}{\overline{\alpha_i}}\bm{h}_i^0-\bm{h}_i^\natural\right\|+\left\|\alpha_i\bm{x}_i^0-\bm{x}_i^\natural\right\|\right\}
\notag\\
\overset{(\text{ii})}{\leq}&\min_{\alpha_i\in\mathbb{ C},|\alpha_i| = 1}\sqrt{s}\kappa\left\{\left\|\alpha_i\bm{h}_i^0-\bm{h}_i^\natural\right\|+\left\|\alpha_i\bm{x}_i^0-\bm{x}_i^\natural\right\|\right\}
\notag\\
\overset{(\text{iii})}
{\leq}&{ C_1}\frac{1}{\log^2m},\label{dist0}
\end{align}
as long as $m\gg (\mu^2+\sigma^2)s\kappa^2K\log^6m$. Here, (i) arises from the inequality that $a^2+b^2\leq (a+b)^2$ for $a,b>0$ and the assumption that $\|\bm{h}_i^\natural\|_2 =\|\bm{x}_i^\natural\|_2$ with $\max_{1\leq i\leq s}\|\bm{x}_i^\natural\|_2=1$, (ii) occurs since the latter optimization problem has strictly smaller feasible set and (iii) derives from Lemma \ref{L6}. With similar strategy, we can get that with high probability 
\begin{align}
\mbox{dist}(\bm{z}^{0,(l)},\bm{z}^\natural)\lesssim \frac{1}{\log^2m},\quad 1\leq l\leq m.
\end{align}
This establishes the inductive hypothesis (\ref{h0}) for $t=0$. We further show the identification of (\ref{h1}) and (\ref{h3}) for $t=0$.
\begin{lemma}\label{L7}
        Suppose that $m\gg (\mu^2+\sigma^2)s^2\kappa^2K\log^3m$. Then with probability at least $1-O(m^{-9})$,
        \begin{align}
        \mathrm{dist}\left(\bm{z}^{0,(l)},\widetilde{\bm{z}}^0\right)&\leq C_2\frac{s\kappa\mu}{\sqrt{m}}\sqrt{\frac{\mu^2sK\log^5m}{{m}}}~\text{and}\label{28}\\
        \max_{1\leq i\leq m}|\bm{b}_l^*\widetilde{\bm{h}}^0_i|\cdot\|\bm{h}_i^\natural\|_2^{-1}&\leq C_4\frac{\mu\log^2m}{\sqrt{m}}.\label{29}
        \end{align}
\end{lemma}
\begin{proof}
        Please refer to Appendix \ref{proof7}.
\end{proof}
\begin{remark}
        Regarding the proof of Lemma \ref{L7}, we decompose $\bm{M}_i$ into the terms $\sum_{j=1}^m\bm{b}_j\bm{b}_j^*\bm{h}_i^\natural\bm{x}_i^{\natural*}\bm{a}_{ij}\bm{a}_{ij}^*$ and $\bm{W}_i = \sum_{j=1}^m\bm{b}_j(\sum_{k\neq i}\bm{b}_j^*\bm{h}_k^\natural\bm{x}_k^{\natural*}\bm{a}_{kj}+e_j)\bm{a}_{ij}^*$. The proof is further facilitated by the Wedin's sin$\Theta$ theorem \cite{dopico2000note} and the bound that with probability $1-O(m^{-9})$ \cite{ling2018regularized}, 
        $
        \|\bm{W}_i\|\leq {(\|\bm{h}_i^\natural\|_2\cdot\|\bm{x}_i^\natural\|_2)}/({2\sqrt{\log m}}),
        $
        provided that $m\gg (\mu^2+\sigma^2)sK\log^2m$.
\end{remark}
Finally, we specify (\ref{h2}) regarding the incoherence of $\bm{x}_0$ with respect to the vector $\bm{a}_{ij}$ for each $1\leq i\leq s, 1\leq j\leq m$.
\begin{lemma}\label{L8}
        Suppose the sample complexity $m\gg (\mu^2+\sigma^2)s^{3/2}K\log^5m$. Then with probability at least $1-O(m^{-9})$,
        \begin{align}
        \max_{1\leq i \leq s, 1\leq j\leq m}\left|\bm{a}_{ij}^*\left(\widetilde{\bm{x}}_i^0 - \bm{x}_i^{\natural}\right)\right|\cdot\|\bm{x}_i^\natural\|_2^{-1}\leq C_3\frac{1}{\sqrt{s}\log^{3/2}m}.
        \end{align}
\end{lemma}
\begin{proof}
The proof follows \cite[Lemma 21]{ma2017implicit}.
\end{proof}

\section{Conclusion}
In this paper, we developed a provable nonconvex demixing procedure from the sum of noisy bilinear measurements via Wirtinger flow without regularization. We demonstrated that, starting with spectral initialization, the iterates of Wirtinger flow keep staying within the region of incoherence and contraction. The restricted strong convexity and qualified level of smoothness of such a region leads to more aggressive step size for gradient descent, thereby significantly accelerating convergence rates. The provable Wirtinger flow algorithm thus can solve the blind demixing problem with regularization free, fast convergence rates with aggressive step size and computational optimality guarantees. Our theoretical analysis are by no means exhaustive, and there are diverse directions that would be of interest for future investigations. For examples, we may leverage provable regularization-free iterates for the constrained nonconvex high-dimensional estimation problems.  Establish optimality for nonconvex estimation problems solved by other regularization-free iterative methods, e.g., the Riemannian optimization algorithms, are also worth being explored.

        \appendices 
         
        \section{Technical Lemmas}\label{TL}
        The following two lemmas, i.e., Lemma \ref{L1.2} and Lemma \ref{L1.1}, are established to proof Lemma \ref{L1}.
        We denote the population Wirtinger Hessian in the noiseless case at the ground truth $\bm{z}^{\natural}$ as
        \begin{align}\label{DF}
        \nabla^2F(\bm{z}^\natural):=\diagg\left(\{\nabla^2_{\bm{z}_i}F\}_{i=1}^s\right),
        \end{align}
        where 
        \[\begin{split}
        \nabla^2_{\bm{z}_i}F:=\left[~\begin{matrix}
        \bm{I}_K &\bm{0}&\bm{0}& \bm{h}_i^\natural \bm{x}_i^{\natural \top}\\
        \bm{0} &  \bm{I}_K& \bm{x}_i^\natural \bm{h}_i^{\natural \top}&\bm{0}\\
        \bm{0}& \big(\bm{x}_i^\natural \bm{h}_i^{\natural \top}\big)^*& \bm{I}_K & \bm{0}\\
        \big(\bm{h}_i^\natural \bm{x}_i^{\natural \top}\big)^*& \bm{0} & \bm{0} & \bm{I}_K
        \end{matrix}~\right] 
        \end{split}\]
        for $i = 1,\cdots, s$.  
        \begin{lemma}\label{L1.2} Recall that $
        	\bm{z} = \left[
        	\bm{z}_1^*\cdots\bm{z}_s^*
        	\right]^*\in \mathbb{ C}^{2sK}~{\textrm{with}}~\bm{z}_i = \left[\bm{h}_i^*~\bm{x}_i^*\right]^*\in \mathbb{ C}^{2K}.$
        	Instate the notations and conditions in the Lemma 1, there are 
        	$
        	\|\nabla^2F(\bm{z}^{\natural}) \|\leq 1+s$
        	and $
        	\bm{u}^*\left[\bm{D}\nabla^2F(\bm{z}^\natural)  + \nabla^2F(\bm{z}^\natural)\bm{D}\right]\bm{u}\geq \frac{1}{\kappa}\|\bm{u}\|_2^2.
        	$
        \end{lemma}
        \begin{proof}
        	Please refer to Appendix \ref{appda} for details.
        \end{proof}
        \begin{lemma}\label{L1.1}
        	Suppose the sample complexity satisfies $m\gg \mu^2s^2\kappa^2K\log^5m$. Then with probability at least $1-O(m^{-10})$, one has
        	$
        	\sup_{\bm{z}\in\mathcal{S}}\|\nabla^2f_{\mathrm{clean}}(\bm{z})-\nabla^2F(\bm{z}^\natural)\|\leq\frac{1}{4},
        	$
        	where the set $\mathcal{S}$ consists of all $\bm{z}$'s satisfying the conditions (\ref{zcondition}) provided in Lemma \ref{L1}.
        \end{lemma}
        \begin{proof}
        	Please refer to Appendix \ref{appdb} for details.
        \end{proof}
        \begin{remark}
        	For the proof of Lemma \ref{L1.2} and Lemma \ref{L1.1}, extension operations are required due to multiple sources in blind demixing. Furthermore, for the proof of Lemma \ref{L1.1}, we decompose the quantity of interest to the sum of spectral norm of random matrix. In particular, the sum of multiple ``incoherence'' signals in (\ref{g1}) and (\ref{g2}) calls for new statistical guarantees for the spectral norm of random matrices over the ``incoherence'' region, which is demonstrated in Lemma \ref{newset} 
        	(see Appendix \ref{TL}) by extending Lemma 59 in \cite{ma2017implicit} for blind deconvolution with single source. 
        \end{remark} 
        
        \begin{lemma}\label{L16}
        	Suppose that $m\gg 1$. The following two bounds hold true.
        	\begin{enumerate}
        		\item If $\left||\alpha_i^t|-1\right|<{1}/{2}$, $i =1,\cdots,s$ and $\mathrm{dist}(\bm{z}^t,\bm{z}^\natural)\leq C_1/\log^2m$, then for $i =1,\cdots,s$
        		\begin{align}
        		\left|\frac{\alpha_i^{t+1}}{\alpha^t_i} - 1\right|\leq c ~\mathrm{dist}(\bm{z}_i^t,\bm{z}_i^\natural)\leq\frac{cC_1}{\log^2m}\label{L526}
        		\end{align}
        		holds  for some absolute constant $c>0$.
        		\item If $\left||\alpha_i^0|-1\right|<{1}/{4}$, $i =1,\cdots,s$ and $\mathrm{dist}(\bm{z}^\tau,\bm{z}^\natural)$ satisfies the condition (\ref{distforT1}) for all $0\leq \tau\leq t$, then for $i =1,\cdots,s$, one has
        		$
        		\left||\alpha_i^{\tau+1}|-1\right|<\frac{1}{2},~0\leq \tau\leq t,
        		$
        		with sufficiently small $C_5>0$  .
        		\begin{proof}
        			The proof follows \cite[Lemma 16]{ma2017implicit}.
        		\end{proof}
        	\end{enumerate}
        \end{lemma}
        We will present that the assumption $\left||\alpha_i^0|-1\right|<{1}/{4}$, for $i =1,\cdots, s$ can be guaranteed with high probability by Lemma \ref{L6}. Based on Lemma \ref{2} and Lemma \ref{L16}, we  conclude that the ratio of consecutive alignment parameters, i.e., $\alpha_i^{t+1}/\alpha_i^t$, $i = 1,\cdots,s$, linearly converges to $1$, and $\alpha_i^t$, $i = 1,\cdots, s$ converges to a point  near to $1$.
        \begin{lemma}\label{L58}
        	
        	Suppose that $\{\bm{A}_{kl}\}_{1\leq l \leq m }$ is a collection of fixed matrices in $\mathbb{C}^{N\times K}.$ For $k\neq i$, we have
        	\begin{align}
        	&\mathbb{P}\left(\left\|\frac{1}{m}\sum_{l=1}^{m}\bm{A}_{kl}\bm{a}_{kl}\bm{a}_{il}^*\right\|\geq 2\theta\left\|\frac{1}{m}\sum_{l=1}^m\bm{A}_{kl}\bm{A}_{kl}^*\right\|^{1/2}\right)\notag\\
        	\leq & \mathrm{exp}\left(c(N+K)-\theta^2m/C\right),~\forall\theta\in(0,1).
        	\end{align}
        	Here, $c,C>0$ are some absolute constants,
        \end{lemma}  
        \begin{proof}
        	For simplicity, we define $\bm{Q} =\sum_{l=1}^{m}\bm{A}_{kl}\bm{a}_{kl}\bm{a}_{il}^*$, where $k\neq i$. We are going to show that 
        	\begin{align}\label{bound2}
        	&\mathbb{P}\left(\frac{1}{m}|\bm{u}^*\bm{Q}\bm{v}|\geq \theta\left\|\frac{1}{m}\sum_{l=1}^m\bm{A}_{kl}\bm{A}_{kl}^*\right\|^{1/2}\right)\notag\\
        	\leq&\mbox{exp}(1-\theta^2m/C),~\forall\theta\in(0,1),
        	\end{align} 
        	holds for any fixed $\bm{u}\in\mathbb{C}^N$, $\bm{v}\in\mathbb{C}^K$ with $\|\bm{u}\|_2 = \|\bm{v}\|_2 = 1$. To achieve this goal, we denote a zero-mean random variable as
        	$
        	w_l = \bm{u}^*\bm{A}_{kl}\bm{a}_{kl}\bm{a}_{il}^*\bm{v},
        	$ where $k\neq i$.
        	Based on the technique provided in  \cite[Lemma 58]{ma2017implicit}, we  accomplish the proof.
        \end{proof} 
        
        The following lemma derives the supremum of the spectral norm of random matrices over an ``incoherence" region. 
        \begin{lemma}\label{newset}
        	Suppose that $\{\bm{A}_{kji}(\bm{h},\bm{x})\}_{1\leq j\leq m}$, where $1\leq k,i\leq s$ and $k\neq i$, is a set of $\mathbb{ C}^{N\times K}$-valued function defined on $\mathbb{C}^{sN}\times \mathbb{C}^{sK}$, such that for all $(\bm{h},\bm{x})$, $(\bm{h}^{\prime},\bm{x}^{\prime})$, $(\bm{h}^{\prime\prime},\bm{h}^{\prime\prime})\in\mathcal{ C}(\frac{\delta}{\kappa\sqrt{s}},\alpha)$ the following conditions hold:
        	$
        	\|\frac{1}{m}\sum_{j=1}^{m}\bm{A}_{kji}(\bm{h},\bm{x})\bm{A}_{kji}^*(\bm{h},\bm{x})\|^{1/2}\leq  M_1$, and $
        	\max_{1\leq j\leq m}\|\bm{A}_{kji}(\bm{h}^{\prime\prime},\bm{x}^{\prime\prime}) - \bm{A}_{kji}(\bm{h}^{\prime},\bm{x}^{\prime})\|
        	\leq M_2$ $\max\{\|\bm{h}_{k}^{\prime\prime} - \bm{h}_{k}^{\prime}\|_2,\|\bm{x}_{k}^{\prime\prime} - \bm{x}_{k}^{\prime}\|_2\}.
        	$
        	
        	Define $P_k(\bm{h},\bm{x}):=\sum_{j=1}^m\bm{A}_{kji}(\bm{h},\bm{x})\bm{a}_{kj}\bm{a}_{ij}^*$, where $k\neq i$. If the parameters $\delta$, $M_1$ and $M_2$ hold that
        	$
        	(\min\{\frac{\delta}{smM_1},1\})^2m\gg (K+N)\log m$ and $ m\gg \kappa\sqrt{s}M_2K,
        	$
        	then with probability exceeding $1-O(m^{-10})$, there is 
        	$
        	\sup_{(\bm{h}_k,\bm{x}_k)\in\mathcal{C}_k(\frac{\delta}{\kappa\sqrt{s}},\alpha)}\|P_k(\bm{h},\bm{x})\|\leq \frac{4\delta}{s}.
        	$
        \end{lemma}    
        \begin{proof}
        	The proof follows the technical method provided in \cite[Lemma 59]{ma2017implicit}.
        \end{proof}
   \section{Proof of Lemma \ref{L1}}\label{proof_L1}
      Combining Lemma \ref{L1.2} and Lemma \ref{L1.1} in Appendix \ref{TL}, we can see that for $\bm{z}\in\mathcal{S}$,
    \begin{align}
    \left\|\nabla^2f_{\mathrm{clean}}(\bm{z})\right\|&\leq\left\|\nabla^2F(\bm{z}^\natural)\right\|+\left\|\nabla^2f_{\mathrm{clean}}(\bm{z})-\nabla^2F(\bm{z}^\natural)\right\|\notag\\
    &\leq 1+s+1/4\leq 2+s,
    \end{align}
    which identifies the upper bound of level of smoothness. We further have
    \begin{align}
    &\bm{u}^*\left[\bm{D}\nabla^2f_{\mathrm{clean}}(\bm{z})+\nabla^2f_{\mathrm{clean}}(\bm{z})\bm{D}\right]\bm{u}\notag\\
    \overset{(\text{i})}{\geq}&\bm{u}^*\left[\bm{D}\nabla^2F(\bm{z}^\natural)+\nabla^2F(\bm{z}^\natural)\bm{D}\right]\bm{u}-2\|\bm{D}\|\cdot\notag\\
    &\left\|\nabla^2f_{\mathrm{clean}}(\bm{z})-\nabla^2F(\bm{z}^\natural)\right\|\|\bm{u}\|_2^2\notag\\
    \overset{(\text{ii})}{\geq}&\frac{1}{\kappa}\|\bm{u}\|_2^2-2(\frac{1}{\kappa}+\frac{\delta}{\kappa\sqrt{s}})\cdot\frac{1}{4}\|\bm{u}\|_2^2\notag\\
    \overset{(\text{iii})}{\geq}&\frac{1}{4\kappa}\|\bm{u}\|_2^2,
    \end{align}
    where $(\mbox{i})$ uses proper reformulation and triangle inequality, $(\mbox{ii})$ is derived from Lemma \ref{L1.1} and the fact that $\|\bm{D}\|\leq \frac{1}{\kappa}+\frac{\delta}{\kappa\sqrt{s}}$, and $(\mbox{iii})$ holds if $\delta\leq \frac{\sqrt{s}}{2}$. Thus, we finish establishing the restricted strong convexity and smoothness in the region of incoherence and contraction.
    \section{Proof of Lemma \ref{L1.2}}
    \label{appda}
    We first provide the  expressions of $
    \bm{C} = \left[~\begin{matrix}
    \bm{C}_1 &\bm{C}_2\\
    \bm{C}_2^*&\bm{C}_3
    \end{matrix}~\right]
 $
    where
    \begin{subequations}
        \begin{align}
        \bm{C}_1& = \sum_{j=1}^{m}|\bm{a}^*_{ij}\bm{x}_i|^2\bm{b}_{j}\bm{b}_j^*,\\
        \bm{C}_2& = \sum_{j=1}^m\left(\sum_{k=1}^{ s}\bm{b}_j^*\left(\bm{h}_k\bm{x}_k^{\ast}-\bm{h}_k^{\natural}\bm{x}_k^{\natural\ast}\right)\bm{a}_{kj}\right)\bm{b}_j\bm{a}_{ij}^*,\\
        \bm{C}_3& = \sum_{j=1}^{m}|\bm{b}^*_{j}\bm{h}_i|^2\bm{a}_{ij}\bm{a}_{ij}^*,
        \end{align}
    \end{subequations}
    and $
    \bm{E} = \left[~\begin{matrix}
    \bm{0 }&\bm{E}_1\\
    \bm{E}_2&\bm{0}
    \end{matrix}~\right]
   $
    where
    \begin{subequations}
        \begin{align}
        \bm{E}_1 &= \sum_{j=1}^m  \bm{b}_{j}\bm{b}_{j}^*\bm{h}_i(\bm{a}_{ij}\bm{a}_{ij}^*\bm{x}_i)^\top,\\
        \bm{E}_2 & = \sum_{j=1}^m  \bm{a}_{ij}\bm{a}_{ij}^*\bm{x}_i(\bm{b}_j\bm{b}_{j}^*\bm{h}_i)^\top.
        \end{align}
    \end{subequations}
    We first prove the identity  $\left\|\nabla^2F(\bm{z}^\natural)\right\| = 1+s$. For $i = 1,\cdots, s$, let $\bm{v}_{i1} = \frac{1}{\sqrt{2}}[~
    \bm{q}~
    \bm{h}_i^\natural~
    \bm{0}~
    \bm{0}~
    \overline{\bm{x}_i^\natural}~
    \bm{w}
    ~]^\top, \bm{v}_{i2} = \frac{1}{\sqrt{2}}[~
    \bm{q}~
    \bm{0}~
    \bm{x}_i^\natural~
    \overline{\bm{h}_i^\natural}\notag~
    \bm{0}~
    \bm{w}
    ~]^\top,\bm{v}_{i3}= \frac{1}{\sqrt{2}}[~
    \bm{q}~
    \bm{h}_i^\natural~
    \bm{0}~
    \bm{0}~
    -\overline{\bm{x}_i^\natural}~
    \bm{w}~]^\top,
    \bm{v}_{i4} = \frac{1}{\sqrt{2}}[~
    \bm{q}~
    \bm{0}~
    \bm{x}_i^\natural~
    -\overline{\bm{h}_i^\natural}~
    \bm{0}~
    \bm{w}~
    ]^\top,
   $
    where $\bm{a}^{\top}$ denote the transpose of the complex vector $\bm{a}$, $\bm{v}_{i1},\bm{v}_{i2},\bm{v}_{i3},\bm{v}_{i4}\in\mathbb{ C}^{4s}$ as well as $\bm{q}\in\mathbb{ R}^{4(i-1)}$ and $\bm{w}\in\mathbb{ R}^{4(s-i)}$ are zero vectors.
    Based on the assumption that $\|\bm{h}_i^\natural\|_2 =\|\bm{x}_i^\natural\|_2$ for $i=1,\cdots,s$, we check that these vectors are from an orthonormal set of size $4s$. Via simple calculations, there is
    $
    \nabla^2 F (\bm{z}^\natural) = \bm{I}_{4sK}+\sum_{i=1}^s (\bm{v}_{i1}\bm{v}_{i1}^* + \bm{v}_{i2}\bm{v}_{i2}^*-\bm{v}_{i3}\bm{v}_{i3}^*-\bm{v}_{i4}\bm{v}_{i4}^*),
   $
    which implies that
    $
    \left\|\nabla^2F(\bm{z}^\natural)\right\| \leq 1+s.
$
    Based on Lemma 26 in \cite{ma2017implicit} and the definition of $\bm{u}_i$ in Lemma 1, for $i = 1,\cdots, s$, there is 
   $
    \bm{u}_i^*\left[\bm{M}_i\nabla_{\bm{z}_i}^2F(\bm{z}^\natural)  + \nabla^2_{\bm{z}_i}F(\bm{z}^\natural)\bm{M}_i\right]\bm{u}_i\geq 1/\kappa\|\bm{u}_i\|_2^2,
  $
    as long as $\delta$ defined in Lemma 1 is small enough, which implies that 
    \begin{align}
    &\bm{u}^*\left[\bm{D}\nabla^2F(\bm{z}^\natural)  + \nabla^2F(\bm{z}^\natural)\bm{D}\right]\bm{u}\notag
    \\=&\sum_{i=1}^{s}\bm{u}_i^*\left[\bm{M}_i\nabla_{\bm{z}_i}^2F(\bm{z}^\natural)  + \nabla^2_{\bm{z}_i}F(\bm{z}^\natural)\bm{M}_i\right]\bm{u}_i\notag
    \\\geq &\frac{1}{\kappa}\sum_{i=1}^{s}\|\bm{u}_i\|_2^2 = \frac{1}{\kappa}\|\bm{u}\|_2^2.
    \end{align}
    
    \section{Proof of Lemma \ref{L1.1}}
    \label{appdb}
    Based on the expression of $\nabla^2f_{\mathrm{clean}}(\bm{z})$ (\ref{chess}) and $\nabla^2F(\bm{z}^\natural)$ (\ref{DF}) and the triangle inequality, we have 
    \begin{align}
    \left\|\nabla^2f_{\mathrm{clean}}(\bm{z})-\nabla^2F(\bm{z}^\natural)\right\|
    \leq \max_{1\leq i\leq s}(\alpha_{i1}+2\alpha_{i2}+4\alpha_{i3}+ 4\alpha_{i4} )
    \end{align}
    where the four terms on the right hand side are defined as follows
    \begin{subequations}
        \begin{align}
        \alpha_{i1} &= \left\| \sum_{j=1}^{m}|\bm{a}^*_{ij}\bm{x}_i|^2\bm{b}_{j}\bm{b}_j^*-\bm{I}_K \right\|,\\
        \alpha_{i2} &= \left\|\sum_{j=1}^{m}|\bm{b}^*_{j}\bm{h}_i|^2\bm{a}_{ij}\bm{a}_{ij}^*-\bm{I}_K \right\|,\\
        \alpha_{i3}& = \left\|  \sum_{j=1}^m\bigg(\sum_{k=1}^{ s}\bm{b}_j^*\left(\bm{h}_k\bm{x}_k^{\ast}-\bm{h}_k^{\natural}\bm{x}_k^{\natural\ast}\right)\bm{a}_{kj}\bigg)\bm{b}_j\bm{a}_{ij}^*  \right\|,\\
        \alpha_{i4}& =  \left\|\sum_{j=1}^m  \bm{b}_{j}\bm{b}_{j}^*\bm{h}_i(\bm{a}_{ij}\bm{a}_{ij}^*\bm{x}_i)^\top-\bm{h}_i^\natural\bm{x}_i^{\natural\top}\right\|.
        \end{align}
    \end{subequations}
    \begin{enumerate}[1)]
        \item   Here, $\alpha_{i1},\alpha_{i2},\alpha_{i4}$ can be bounded through \cite[Lemma 27]{ma2017implicit}. In particular, with probability $1-O(m^{-10})$,
        \begin{align}
        \max_{1\leq i \leq s}\sup_{\bm{z}\in\mathcal{S}}\alpha_{i1}\lesssim\sqrt{\frac{K}{m}\log m} + C_3\frac{1}{\log m}.
        \end{align}
        In addition, with probability at least $1-O(m^{-10})$, we have 
        \begin{align}
        \max_{1\leq i \leq s}\sup_{\bm{z}\in\mathcal{S}} \alpha_{i2}&\leq 7\frac{\delta}{\kappa\sqrt{s}},\\
        \max_{1\leq i \leq s}\sup_{\bm{z}\in\mathcal{S}} \alpha_{i4}&\leq 11\frac{\delta}{\kappa\sqrt{s}}
        \end{align}
        as long as $m\gg (\mu^2/\delta)s\kappa^2K\log^5m$.

        \item To control $\alpha_{i3}$, similar to the set defined in \cite{ma2017implicit}, we define a new set for $(\bm{h},\bm{x})\in\mathbb{ C}^{sK}\times \mathbb{ C}^{sK}$
        \begin{align}
        \mathcal{C}(\xi,\zeta): = \bigg\{&(\bm{h},\bm{x}):\max_{1\leq i \leq s}\max\left\{\|\bm{h}_i-\bm{h}_i^\natural\|_2,\|\bm{x}_i-\bm{x}_i^\natural\|_2\right\}\notag\\&\leq \xi~\notag\text{and} \max_{1\leq i \leq s,1\leq j \leq m}\left|\bm{b}_{j}^*{\bm{h}}_i\right|\cdot\|\bm{h}_i^\natural\|
        _2\leq \frac{\zeta}{\sqrt{m}}\bigg\},
        \end{align}
        where $\bm{h}$ is composed of $\bm{h}_1,\cdots,\bm{h}_s$ and $\bm{x}$ is composed of $\bm{x}_1,\cdots,\bm{x}_s$.
        Note that the set $\mathcal{S}$ defined in Lemma \ref{L1.1} satisfies $\mathcal{S}\subseteq  \mathcal{C}(\frac{\delta}{\kappa\sqrt{s}},2C_4\mu\log^2m)$, thus it suffices to specify $\sup_{\bm{z}\in\mathcal{C}(\frac{\delta}{\kappa\sqrt{s}},2C_4\mu\log^2m)} \alpha_{i3}$ in order to control $\alpha_{i3}$. We are going to exploit Lemma 59 in \cite{ma2017implicit} to derive that with probability at least $1-O(m^{-10})$ 
        \begin{align}\label{a3}
        \max_{1\leq i \leq s}\sup_{\bm{z}\in\mathcal{C}(\frac{\delta}{\kappa\sqrt{s}},2C_4\mu\log^2m)}\alpha_{i3}\leq 4\delta +\frac{7\delta}{\kappa\sqrt{s}}.
        \end{align}
        To achieve this goal, we define 
        $
        \bm{\Delta}_{ij}(\bm{h},\bm{x}):=\sum_{k\neq i}\left(\bm{h}_k\bm{x}^*_k - \bm{h}^\natural_k\bm{x}^{\natural*}_k
        \right)\bm{a}_{kj}
        $
        and 
    $
        R_i(\bm{\bm{h},\bm{x}}):=R_{i,\text{clean}}(\bm{h,\bm{x}})+\sum_{j=1}^m\bm{b}_j\bm{b}_j^*\bm{\Delta}_{ij}(\bm{h},\bm{x})\bm{a}_{ij}^*,
   $
        where the first term is denoted as $R_{i,\text{clean}}(\bm{h,\bm{x}}) =\sum_{j=1}^{m}\bm{b}_j\bm{b}_j^*(\bm{h}_i\bm{x}_i^*-\bm{h}_i^\natural\bm{x}_i^{\natural*})\bm{a}_{ij}\bm{a}_{ij}^*$.         
        The original inequality (\ref{a3}) can be represented as
        \begin{align}
        &\mathbb{P}\left(\sup_{\bm{z}\in\mathcal{C}(\frac{\delta}{\kappa\sqrt{s}},2C_4\mu\log^2m)}\left\|R_i(\bm{h},\bm{x})\right\|\geq 4\delta+7\frac{\delta}{\kappa\sqrt{s}}\right)\notag\\\lesssim& m^{-10}.
        \end{align}
        Note that one has $\mathbb{ E}[R_{i,\text{clean}}(\bm{h},\bm{x})] = \bm{h}_i\bm{x}_i^*-\bm{h}_i^\natural\bm{x}_i^{\natural*}$ and the spectral norm is bounded by $\|\mathbb{ E}[R_{i,\text{clean}}(\bm{h},\bm{x})]\|\leq {3\delta}/(\kappa{\sqrt{s}})$  \cite[Section C.1.2]{ma2017implicit} when $\bm{h}_i,\bm{x}_i$ are fixed. Based on the conclusion provided in \cite[Section C.1.2]{ma2017implicit}, for $(\bm{h},\bm{x})\in\mathcal{C}(\frac{\delta}{\kappa\sqrt{s}},2C_4\mu\log^2m)$, it yields
        \begin{align}
        &\mathbb{P}\left(\sup_{(\bm{h},\bm{x})}\left\|R_{i,\text{clean}}(\bm{h},\bm{x})-\mathbb{ E}[R_{i,\text{clean}}(\bm{h},\bm{x})]\right\|\geq 4\frac{\delta}{\kappa\sqrt{s}}\right)\notag\\\lesssim& ~m^{-10},
        \end{align}
        as long as $m\gg (\mu^2/\delta^2)s\kappa^2K\log^5m$.
        It thus suffices to show that 
        \begin{align}
        \mathbb{P}\left(\sup_{(\bm{h},\bm{x})}\left\|\sum_{j=1}^m\bm{b}_j\bm{b}_j^*\bm{\Delta}_{ij}(\bm{h},\bm{x})\bm{a}_{ij}^*\right\|\geq 4\delta\right)\lesssim m^{-10},
        \end{align}
        where ${(\bm{h},\bm{x})\in\mathcal{C}(\frac{\delta}{\kappa\sqrt{s}},2C_4\mu\log^2m)}$. We are positioned to invoke Lemma \ref{newset} to achieve the above result. Specifically, let $\bm{A}_{kji}(\bm{h},\bm{x}) = \bm{b}_j\bm{b}_j^*\bm{\Gamma}_{ik}$ where $\bm{\Gamma}_{ik} = \bm{h}_k\bm{x}_k^*-\bm{h}_k^\natural\bm{x}_k^{\natural*}$ with $k\neq i$. We further define 
        $
        \tau = \arg\max_{1\leq k\leq s,k\neq i }\|\bm{A}_{kji}(\bm{h},\bm{x})\bm{a}_{kj}\bm{a}_{ij}^*\|.
    $
        Hence, it suffices to show that 
        \begin{align}\label{55}
        \mathbb{P}\left(\sup_{(\bm{h},\bm{x})}\left\|\sum_{j=1}^m\bm{A}_{\tau ji}(\bm{h},\bm{x})\bm{a}_{\tau j}\bm{a}_{ij}^*\right\|\geq \frac{4\delta}{s}\right)\lesssim m^{-10},
        \end{align}
        By choosing $M_1\leq5C_4\mu\log^2m/m$ and $M_2\leq  4K/m$, we invoke Lemma \ref{newset} and finish the proof of inequality (\ref{55}).
        
        \item Based on the previous bounds, we deduced that with probability $1-O(m^{-10})$,
        \begin{align}
        &\left\|\nabla^2f_{\mathrm{clean}}(\bm{z})-\nabla^2F(\bm{z}^\natural)\right\|\notag\\
        \lesssim&\left(\sqrt{\frac{K}{m}\log m} + C_3\frac{1}{\log m}\right)+\delta\leq \frac{1}{4},
        \end{align}
        as long as $\delta>0$ is a small constant and $m\gg \mu^2s^2\kappa^2K\log^5m$, as desired.
    \end{enumerate}
    \section{Proof of Lemma \ref{2}}\label{proof2}
    Based on the definition of $\alpha^{t+1}_k(\ref{a}),~ k =1,\cdots,s$ , one has 
    \begin{align}
    &\mbox{dist}^2\left(\bm{z}^{t+1},\bm{z}^\natural\right)\leq  \sum_{k=1}^s\mbox{dist}^2\left(\bm{z}_k^{t+1},\bm{z}_k^\natural\right) \notag\\
    \leq & s\kappa^2\left\|\frac{1}{\overline{\alpha_{k}^{t+1}}}\bm{h}_k^{t}  - \bm{h}_k^\natural\right\|_2^2+s\kappa^2\left\|{{\alpha_{k}^{t}}\bm{x}_k^{t+1}}  - \bm{x}_k^\natural\right\|_2^2.
    \end{align}

    By denoting $\widetilde{\bm{h}}_k^{t} = \frac{1}{\overline{\alpha_{k}^{t}}}\bm{h}_k^{t},~\widetilde{\bm{x}}_k^t = \alpha_k^t\bm{x}_k^{t},~\widehat{\bm{h}}_k^{t+1} =\frac{1}{\overline{\alpha_{k}^{t}}}\bm{h}_k^{t+1} $ and $ \widehat{\bm{x}}_k^{t+1} = \alpha_k^t\bm{x}_k^{t+1}$, we have
    \begin{align}\label{grad}
    \left[~\begin{matrix}
    \widehat{\bm{h}}_k^{t+1}-\bm{h}_k^\natural\\
    \widehat{\bm{x}}_k^{t+1}-\bm{x}_k^\natural\\
    \overline{\widehat{\bm{h}}_k^{t+1}-\bm{h}_k^\natural}\\
    \overline{\widehat{\bm{x}}_k^{t+1}-\bm{x}_k^\natural}
    \end{matrix}~\right] &= \left[~\begin{matrix}
    \widetilde{\bm{h}}_k^{t}-\bm{h}_k^\natural\\
    \widetilde{\bm{x}}_k^{t}-\bm{x}_k^\natural\\
    \overline{\widetilde{\bm{h}}_k^{t}-\bm{h}_k^\natural}\\
    \overline{\widetilde{\bm{x}}_k^{t}-\bm{x}_k^\natural}
    \end{matrix}~\right]-\eta\bm{W}_k\left[~\begin{matrix}
    \nabla_{\bm{h}_k}f(\widetilde{\bm{z}}^t) \\
    \nabla_{\bm{x}_k}f(\widetilde{\bm{z}}^t) \\
    \overline{ \nabla_{\bm{h}_k}f(\widetilde{\bm{z}}^t) }\\
    \overline{ \nabla_{\bm{x}_k}f(\widetilde{\bm{z}}^t) }
    \end{matrix}~\right],
    \end{align}
    and $
    \bm{W}_k = \diagg\big(\big[
    \|\widetilde{\bm{x}}^t_k\|_2^{-2}\bm{I}_K,
    \|\widetilde{\bm{h}}^t_k\|_2^{-2}\bm{I}_K,
   \|\widetilde{\bm{x}}^t_k\|_2^{-2}\bm{I}_K,
    \|\widetilde{\bm{h}}^t_k\|_2^{-2}$ $\bm{I}_K
    \big]    \big).
   $
    According to the fundamental theorem of calculus provided in \cite{ma2017implicit} together with the definition of the noiseless objective function $f_{\text{clean}}$ and the noiseless Wirtinger Hessian $\nabla^2_{\bm{z}_k}f_{\text{clean}}$ (\ref{chess}), we get
    \begin{align}\label{grads}
    \left[~\begin{matrix}
    \nabla_{\bm{h}_k}f(\widetilde{\bm{z}}^t) \\
    \nabla_{\bm{x}_k}f(\widetilde{\bm{z}}^t) \\
    \overline{ \nabla_{\bm{h}_k}f(\widetilde{\bm{z}}^t) }\\
    \overline{ \nabla_{\bm{x}_i}f(\widetilde{\bm{z}}^t) }
    \end{matrix}~\right] &= \left[~\begin{matrix}
    \nabla_{\bm{h}_k}f_{\text{clean}}(\widetilde{\bm{z}}^t) \\
    \nabla_{\bm{x}_k}f_{\text{clean}}(\widetilde{\bm{z}}^t) \\
    \overline{ \nabla_{\bm{h}_k}f_{\text{clean}}(\widetilde{\bm{z}}^t) }\\
    \overline{ \nabla_{\bm{x}_k}f_{\text{clean}}(\widetilde{\bm{z}}^t) }
    \end{matrix}~\right] + \left[~\begin{matrix}
    \mathcal{A}_k(\bm{e})\bm{x}_k^t  \\
    \mathcal{A}_k^*(\bm{e})\bm{h}_k^t \\
    \overline{ \mathcal{A}_k(\bm{e})\bm{x}_k^t }\\
    \overline{ \mathcal{A}_k^*(\bm{e})\bm{h}_k^t }
    \end{matrix}~\right]\notag\\
    & = \bm{H}_k \left[~\begin{matrix}
    \widetilde{\bm{h}}_k^{t}-\bm{h}_k^\natural\\
    \widetilde{\bm{x}}_k^{t}-\bm{x}_k^\natural\\
    \overline{\widetilde{\bm{h}}_k^{t}-\bm{h}_k^\natural}\\
    \overline{\widetilde{\bm{x}}_k^{t}-\bm{x}_k^\natural}
    \end{matrix}~\right]+ \left[~\begin{matrix}
    \mathcal{A}_k(\bm{e})\bm{x}_k^t  \\
    \mathcal{A}_k^*(\bm{e})\bm{h}_k^t \\
    \overline{ \mathcal{A}_k(\bm{e})\bm{x}_k^t }\\
    \overline{ \mathcal{A}_k^*(\bm{e})\bm{h}_k^t }
    \end{matrix}~\right],
    \end{align}
    where $\bm{H}_k = \int_{0}^{1}\nabla^2_{\bm{z}_k}f_{\text{clean}}\left(\bm{z}(\tau)\right)d \tau$ with $\bm{z}(\tau):= \bm{z}^\natural+\tau\left(\widetilde{\bm{z}}^t - \bm{z}^\natural\right)$ and $\mathcal{A}_k(\bm{e})= \sum_{j=1}^m{e}_j\bm{b}_j\bm{a}_{kj}^*$ and $\mathcal{A}_k^*(\bm{e})= \sum_{j=1}^m\overline{{e}_j}\bm{a}_{kj}\bm{b}_j^*$. Since $\bm{z}(\tau)$ lies between $\widetilde{\bm{z}}^t$ and $\bm{z}^\natural$, for all $\tau\in[0,1]$, $\bm{z}(\tau)$ satisfies the assumption (\ref{con85}).

    For simplicity, we denote $\widehat{\bm{z}}_k^{t+1} = [\widehat{\bm{h}}_k^{t+1\ast} ~ \widehat{\bm{x}}_k^{t+1\ast}
    ]^\ast$. Substituting (\ref{grads}) to (\ref{grad}), one has
    \begin{align}\label{ida}
    \left[~\begin{matrix}
    \widehat{\bm{z}}_k^{t+1}-\bm{z}_k^\natural\\
    \overline{\widehat{\bm{z}}_k^{t+1}-\bm{z}_k^\natural}
    \end{matrix}~\right] = \bm{\varphi}_k^{t}+\bm{\psi}_k^{t},
    \end{align}
    where 
    \begin{align}
    \bm{\varphi}_k^t = (\bm{I} - \eta\bm{W}_k\bm{H}_k)\left[~\begin{matrix}
    \widetilde{\bm{z}}_k^{t}-\bm{z}_k^\natural\\
    \overline{\widetilde{\bm{z}}_k^{t}-\bm{z}_k^\natural}
    \end{matrix}~\right],\bm{\psi}_k^t = \left[~\begin{matrix}
    \mathcal{A}_k(\bm{e})\bm{x}_k^t  \\
    \mathcal{A}_k^*(\bm{e})\bm{h}_k^t \\
    \overline{ \mathcal{A}_k(\bm{e})\bm{x}_k^t }\\
    \overline{ \mathcal{A}_k^*(\bm{e})\bm{h}_k^t }
    \end{matrix}~\right].\notag
    \end{align} 
    Take the Euclidean norm of both sides of (\ref{ida}) to arrive
    \begin{align}\label{mme}
    \|\bm{\varphi}_{k}^t+\bm{\psi}_k^t\|_2 \leq \|\bm{\varphi}_k^t\|_2 + \|\bm{\psi}_k^t\|_2.
    \end{align}
    We first control the second Euclidean norm at the right-hand side of the equation (\ref{mme}):
    $
    \|\bm{\psi}_k^t\|_2^2 
    \leq 16\left\|\mathcal{A}_k(\bm{e})\right\|^2,
   $
    where we use the fact that $\max\{\|\bm{x}_k\|_2,\|\bm{h}_k\|_2\}\leq 2$ for $1\leq k\leq s$. Based on the paper \cite[Section C.2]{ma2017implicit}, the squared Euclidean norm of $\bm{\varphi}_k^t$ is bounded by
   $
    \|\bm{\varphi}_k^t\|_2^2\leq 2(1-\eta/(8\kappa))\|\widetilde{\bm{z}}_k^t - \bm{z}_k^\natural\|_2^2,
  $
    under the assumption (\ref{con85}).
    We thus conclude that 
    \begin{align}
    \|\bm{\varphi}_{k}^t+\bm{\psi}_k^t\|_2 \leq \sqrt{2}(1-\eta/(8\kappa))^{1/2}\|\widetilde{\bm{z}}_k^t - \bm{z}_k^\natural\|_2+4\left\|\mathcal{A}_k(\bm{e})\right\|,
    \end{align}
    and hence 
    \begin{align}\label{distz}
    \|\widetilde{\bm{z}}_k^{t+1}-\bm{z}_k^\natural\|_2&\leq\|\widehat{\bm{z}}_k^{t+1}-\bm{z}_k^\natural\|_2\leq\sqrt{2}/2\|\bm{\varphi}_{k}^t+\bm{\psi}_k^t\|_2\notag\\
    &\leq (1-\eta/(16\kappa))\|\widetilde{\bm{z}}_k^t - \bm{z}_k^\natural\|_2+3\left\|\mathcal{A}_k(\bm{e})\right\|.
    \end{align}
    Integrating the above inequality (\ref{distz}) for $i = 1,\cdots, s$, we further obtain 
  $
    \mathrm{dist}(\bm{z}^{t+1},\bm{z}^\natural)\leq (1-\eta/(16\kappa))\mathrm{dist}(\bm{z}^t,\bm{z}^\natural) + 3\sqrt{s}\kappa\max_{1\leq k \leq s}\left
    \|\mathcal{A}_k(\bm{e})\right\|.
  $

    \section{Proof of Lemma \ref{3}}\label{proof4}
    Define the alignment parameter between ${\bm{z}}_i^{t,(l)} = [{\bm{h}}_i^{t,(l)*}~~{\bm{x}}_i^{t,(l)*}]^*$ and $\widetilde{\bm{z}}_i^t = [\widetilde{\bm{h}}_i^{t*}~~\widetilde{\bm{x}}_i^{t*}]^*$ as
    \begin{align}
    \alpha_{i,\text{mutual}}^{t,(l)}: = \argmin_{\alpha\in\mathbb{C}}\left\| \frac{1}{\overline{\alpha}}\bm{h}_i^{t,(l)}-\frac{1}{\overline{\alpha_i^t}}\bm{h}_i^{t}\right\|_2^2+ \left\|\alpha\bm{x}_i^{t,(l)}-\alpha_i^t\bm{x}_i^{t}\right\|_2^2,
    \end{align}
    where $\widetilde{\bm{h}}^t_i = \frac{1}{\overline{\alpha_i^t}}\bm{h}_i^t$ and $\widetilde{\bm{x}} ^t_i= \alpha_i^t\bm{x}_i^t$ for $i = 1,\cdots, s$. In addition, we denote $\widehat{\bm{z}}_i^{t,(l)} = [\widehat{\bm{h}}_i^{t,(l)*}~\widehat{\bm{x}}_i^{t,(l)*}]^*$ where
    \begin{align}
    \widehat{\bm{h}}_i^{t,(l)}:=\frac{1}{\overline{\alpha_{i,\text{mutual}}^{t,(l)}}}\bm{h}_i^{t,(l)}~\text{and}~\bm{x}_i^{t,(l)}:={{\alpha_{i,\text{mutual}}^{t,(l)}}}\bm{x}_i^{t,(l)}.
    \end{align}
    In view of the above notions and technical methods in \cite[Section C.3]{ma2017implicit}, we have
    \begin{align}\label{78}
    \mathrm{dist}\left(\bm{z}^{t+1,(l)},\widetilde{\bm{z}}^{t+1}\right)
    \leq \kappa\sqrt{\sum_{k = 1}^s \max\left\{\left|\frac{\alpha_i^{t+1}}{\alpha_i^t}\right|,\left|\frac{\alpha_i^{t}}{\alpha_i^{t+1}}\right|\right\}^2\|\bm{J}_k\|^2},
    \end{align}
    where
    $
    \bm{J}_k = \left[\begin{matrix}
    \frac{1}{\overline{\alpha_{k,\text{mutual}}^{t,(l)}}}\bm{h}_k^{t+1,(l)} - \frac{1}{\overline{\alpha_k^t}}\bm{h}_k^{t+1}\\{{\alpha_{k,\text{mutual}}^{t,(l)}}}\bm{x}_k^{t+1,(l)} - {{\alpha_k^t}}\bm{x}_k^{t+1}
    \end{matrix}
    \right].
    $
    By further applying the update rule in Algorithm \ref{algo}, we get
    \begin{align}
    \bm{J}_{k} = \left[\begin{matrix}
    \widehat{\bm{h}}_k^{t,(l)} -\frac{\eta}{\|\widehat{\bm{x}}_k^{t,(l)}\|_2^2} \nabla_{\bm{h}_k}f^{(l)}(\widehat{\bm{h}}^{t,(l)},\widehat{\bm{x}}^{t,(l)}) - \bm{U}_k\\
    \widehat{\bm{x}}_k^{t,(l)} -\frac{\eta}{\|\widehat{\bm{h}}_k^{t,(l)}\|_2^2} \nabla_{\bm{x}_k}f^{(l)}(\widehat{\bm{h}}^{t,(l)},\widehat{\bm{x}}^{t,(l)}) - \bm{V}_k
    \end{matrix}\right]
    \end{align}
    where $\nabla_{\bm{h}_k}f^{(l)}(\bm{h},\bm{x})$ and $\nabla_{\bm{x}_k}f^{(l)}(\bm{h},\bm{x})$ are defined as
    \begin{align}
    \nabla_{\bm{h}_k}f^{(l)}(\bm{h},\bm{x})& = \nabla_{\bm{h}_k}f(\bm{h},\bm{x})-\bm{R}_l\bm{b}_l\bm{a}_{kl}^*\bm{x}_k,\notag\\
    \nabla_{\bm{x}_k}f^{(l)}(\bm{h},\bm{x})& = \nabla_{\bm{x}_k}f(\bm{h},\bm{x})-\overline{\bm{R}_l}\bm{a}_{kl}\bm{b}_l^*\bm{h}_k,\notag
    \end{align}
    with
    $
    \bm{R}_l = \sum_{i=1}^{ s}\bm{b}_l^*\bm{h}_i\bm{x}_i^{\ast}\bm{a}_{il}-{y}_l,  
   $
    and
    \begin{align}
    \bm{U}_k& = \widetilde{\bm{h}}_k^{t} -\frac{\eta}{\|\widetilde{\bm{x}}_k^{t}\|_2^2} \nabla_{\bm{h}_k}f( \widetilde{\bm{h}}^{t}, \widetilde{\bm{x}}^{t}),\notag\\
    \bm{V}_k& = \widetilde{\bm{x}}_k^{t} -\frac{\eta}{\|\widetilde{\bm{h}}_k^{t}\|_2^2} \nabla_{\bm{x}_k}f( \widetilde{\bm{h}}^{t}, \widetilde{\bm{x}}^{t}).\notag
    \end{align}
    Inspired by \cite[Section C.3]{ma2017implicit}, by further derivation, we obtain 
    \begin{align}\label{81}
    \bm{J}_k = \bm{J}_{k1}+\eta\bm{J}_{k2}-\eta\bm{J}_{k3},
    \end{align}
    where
    \begin{align}
    \bm{J}_{k1} = &\left[\begin{matrix}
    \widehat{\bm{h}}_k^{t,(l)} -\frac{\eta}{\|\widehat{\bm{x}}_k^{t,(l)}\|_2^2} \nabla_{\bm{h}_k}f(\widehat{\bm{h}}^{t,(l)},\widehat{\bm{x}}^{t,(l)}) \\
    \widehat{\bm{x}}_k^{t,(l)} -\frac{\eta}{\|\widehat{\bm{h}}_k^{t,(l)}\|_2^2} \nabla_{\bm{x}_k}f(\widehat{\bm{h}}^{t,(l)},\widehat{\bm{x}}^{t,(l)}) 
    \end{matrix}\right]-\notag\\
    &\left[\begin{matrix}
    \widetilde{\bm{h}}_k^{t} -\frac{\eta}{\|\widehat{\bm{x}}_k^{t,(l)}\|_2^2} \nabla_{\bm{h}_k}f(\widetilde{\bm{h}}^{t},\widetilde{\bm{x}}^{t}) \\
    \widetilde{\bm{x}}_k^{t} -\frac{\eta}{\|\widehat{\bm{h}}_k^{t,(l)}\|_2^2} \nabla_{\bm{x}_k}f(\widetilde{\bm{h}}^{t},\widetilde{\bm{x}}^{t}) 
    \end{matrix}\right],\notag\\
    \bm{J}_{k2} = &\left[\begin{matrix}
    \Big(\frac{1}{\|\widetilde{\bm{x}}_k^{t}\|_2^2}-\frac{1}{\|\widehat{\bm{x}}_k^{t,(l)}\|_2^2}\Big) \nabla_{\bm{h}_k}f(\widetilde{\bm{h}}^{t},\widetilde{\bm{x}}^{t}) \\
    \Big(\frac{1}{\|\widetilde{\bm{h}}_k^{t}\|_2^2}-\frac{1}{\|\widehat{\bm{h}}_k^{t,(l)}\|_2^2}\Big) \nabla_{\bm{x}_k}f(\widetilde{\bm{h}}^{t},\widetilde{\bm{x}}^{t}) 
    \end{matrix}\right],\notag\\
    \bm{J}_{k3} = &\left[\begin{matrix}
    \frac{1}{\|\widehat{\bm{x}}_k^{t,(l)}\|_2^2}\left(\sum_{i=1}^{ s}\bm{b}_l^*\widehat{\bm{h}}_i^{t,(l)}\widehat{\bm{x}}_i^{t,(l)\ast}\bm{a}_{il}-{y}_l\right)\bm{b}_l\bm{a}_{kl}^*\widehat{\bm{x}}_k^{t,(l)}\\
    \frac{1}{\|\widehat{\bm{h}}_k^{t,(l)}\|_2^2}\overline{\left(\sum_{i=1}^{ s}\bm{b}_l^*\widehat{\bm{h}}_i^{t,(l)}\widehat{\bm{x}}_i^{t,(l)\ast}\bm{a}_{il}-{y}_l\right)}\bm{a}_{kl}\bm{b}_l^*\widehat{\bm{h}}_k^{t,(l)}\notag
    \end{matrix}\right].
    \end{align}
    We shall control the three terms $\bm{J}_{k1}$, $\bm{J}_{k2}$ and $\bm{J}_{k3}$.
    \begin{enumerate}[1.]
        \item In terms of the first term $\bm{J}_{k1}$, we can exploit the same strategy as in Appendix \ref{proof2} and conclude that 
        \begin{align}\label{j1}
        \|\bm{J}_{k1}\|\leq \left(1-\frac{\eta}{16\kappa}+C_6\frac{1}{\log^2m}\right)\|\widehat{\bm{z}}_k^{t,(l)}-\widetilde{\bm{z}}_k^{t}\|_2,
        \end{align}
        provided that $m\gg (\mu^2+\sigma^2)s\kappa K\log^{13/2}m$ for the constant $C_6>0$.
        \item  Regarding to the second term $J_2$, based on \cite[Appendix C.3]{ma2017implicit} and the bound on $\left\|\mathcal{ A}_k(\bm{e})\right\|$ provided in \cite[Section 6.5]{ma2017implicit} that
        with probability at least $1-O(m^{-9})$, there holds
        $
        \max_{1\leq i \leq s}\left\|\mathcal{ A}_i(\bm{e})\right\|\leq C_0\sigma \sqrt{\frac{ 10sK\log^2 m }{m}},\notag
    $
        for some absolute constant $C_0>0$ and $\sigma$ is defined in Section \ref{form}, it yields that
        \begin{align}\label{j2}
        \|\bm{J}_2\|_2\lesssim C_7\frac{1}{\log^2m}\|\widehat{\bm{z}}_k^{t,(l)}-\widetilde{\bm{z}}_k^{t}\|_2.
        \end{align}
        \item In terms of the last term $\bm{J}_{k3}$, 
        based on the technical method used in \cite[Appendix C.3]{ma2017implicit} and the fact that $|e_j| \leq \sigma^2/m\ll 1$, we get 
        \begin{align}\label{j3}
        \|\bm{J}_{k3}\|_2\lesssim(C_4)^2\frac{\mu}{\sqrt{m}}\sqrt{\frac{\mu^2sK\log^9m}{m}},
        \end{align}
        provided that $m\gg (\mu+\sigma^2)s^2\kappa K\log^{5/2}m$.
    \end{enumerate}
    Combining the bounds (\ref{78}), (\ref{j1}), (\ref{j2}), (\ref{j3}) and the equation (\ref{81}), there exist a constant $C>0$ such that
    \begin{align}\label{last}
    &\mathrm{dist}\left(\bm{z}^{t+1,(l)},\widetilde{\bm{z}}^{t+1}\right)\notag\\
    \leq& \sqrt{s}\kappa\max\left\{\left|\frac{\alpha_i^{t+1}}{\alpha_i^t}\right|,\left|\frac{\alpha_i^{t}}{\alpha_i^{t+1}}\right|\right\}\Bigg\{\bigg(1-\frac{\eta}{16\kappa}+\frac{C_6}{\log^2m}+\frac{CC_7\eta}{\log^2m}\bigg)\cdot\notag\\
    &~~~~\bigg\|\widehat{\bm{z}}_k^{t,(l)}-\widetilde{\bm{z}}_k^{t}\bigg\|_2 + C(C_4)^2\eta\frac{\mu}{\sqrt{m}}\sqrt{\frac{\mu^2sK\log^9m}{m}}\Bigg\}\notag\\
     \leq& C_2\frac{s\kappa\mu}{\sqrt{m}}\sqrt{\frac{\mu^2K\log^9m}{m}},
    \end{align}
    with $m\gg (\mu^2+\sigma^2)s^{2}\kappa K\log^{13/2}m
    $, $C_2\gg (C_4)^2$ and the bound that
    $
    \max\{|{\alpha_i^{t+1}}/{\alpha_i^t}|,|{\alpha_i^{t}}/{\alpha_i^{t+1}}|\}\leq \frac{1-\eta/(21\kappa)}{1-\eta/(20\kappa)}
    $
    which is derived from Lemma \ref{L16}.
    Hence the inequality (\ref{last}) verifies the induction hypothesis (\ref{h1}) at $(t+1)$-iterate with sufficiently large $C_2$ and sufficiently large $m$.
    
    Finally, we establish the second claim in the lemma based on the technical methods in \cite[Section C.3]{ma2017implicit} and the induction hypothesis (\ref{h11}), we deduced that 
    \begin{align}
    \left\|\widetilde{\bm{z}}^{t+1,(l)}-\widetilde{\bm{z}}^{t+1}\right\|_2&\lesssim\left\|\widehat{\bm{z}}^{t+1,(l)}-\widetilde{\bm{z}}^{t+1}\right\|_2\notag\\
    &\lesssim C_2\frac{s\mu}{\sqrt{m}}\sqrt{\frac{\mu^2K\log^9m}{m}}.
    \end{align}

    \section{Proof of Lemma \ref{4}}\label{proof5}
    Similar to the strategy used in \cite[Section C.4]{ma2017implicit}, it suffices to control $|\bm{b}_l^*\frac{1}{\overline{\alpha_i^t}}\bm{h}_i^{t+1}|$ to finish the proof, as
    \begin{align}
    &\max_{1\leq i \leq s, 1\leq l\leq m}\left|\bm{b}_l^*\frac{1}{\overline{\alpha_{i}^{t+1}}}\bm{h}_i^{t+1}\right|\cdot\|\bm{h}_i^\natural\|_2^{-1}
    \notag\\
    \leq&  (1+\delta)\left|\bm{b}_l^*\frac{1}{\overline{\alpha_{i}^{t}}}\bm{h}_i^{t+1}\right|\cdot\|\bm{h}_i^\natural\|_2^{-1}
    \end{align}
    for some small $\delta\asymp 1/\log^2m$.
    The gradient update rule for $\bm{h}_i^{t+1}$ is written as
    \begin{align}\label{107}
    \frac{1}{\overline{\alpha_i^t}}\bm{h}_i^{t+1} = &\widetilde{\bm{h}}^t_i-\eta\xi_i \sum_{j=1}^m\sum_{k=1}^s\bm{b}_j\bm{b}_j^*(\widetilde{\bm{h}}^t_k\widetilde{\bm{x}}^{t*}_k - \bm{h}_k^\natural\bm{h}_k^{\natural*})\bm{a}_{kj}\bm{a}_{ij}^*\widetilde{\bm{x}}^t_i\notag\\
    &+\eta\xi_i\sum_{j=1}^m e_j\bm{b}_j\bm{a}_{ij}^*\widetilde{\bm{x}}^t_i,
    \end{align}
    where $\xi_i = \frac{1}{\|\widetilde{\bm{x}}_i^t\|_2^2}$ and $\widetilde{\bm{h}}^t_i = \frac{1}{\overline{\alpha_i^t}}\bm{h}_i^t$ and $\widetilde{\bm{x}} ^t_i= \alpha_i^t\bm{x}_i^t$ for $i = 1,\cdots, s$. The formula (\ref{107}) can be further decomposed into the following terms
    \begin{align}
    &\frac{1}{\overline{\alpha_i^t}}\bm{h}_i^{t+1} = \widetilde{\bm{h}}^t_i -\eta\xi_i\sum_{j=1}^m\sum_{k=1}^s\bm{b}_j\bm{b}_j^*\widetilde{\bm{h}}_k^t\widetilde{\bm{x}}_k^{t*}\bm{a}_{kj}\bm{a}_{ij}^*\widetilde{\bm{x}}_i^t +\notag\\
    &\eta\xi_i\sum_{j=1}^m\sum_{k=1}^s\bm{b}_j\bm{b}_j^*{\bm{h}}_k^\natural{\bm{x}}_k^{\natural*}\bm{a}_{kj}\bm{a}_{ij}^*\widetilde{\bm{x}}_i^t+\eta\xi_i\sum_{j=1}^m e_j\bm{b}_j\bm{a}_{ij}^*\widetilde{\bm{x}}^t_i\notag\\
    =& \widetilde{\bm{h}}_i^t-\eta\xi_i\sum_{k=1}^s\widetilde{\bm{h}}_k^t\|\bm{x}_k^\natural\|_2^2-\eta\xi_i\bm{v}_{i1}-\eta\xi_i\bm{v}_{i2}+\eta\xi_i\bm{v}_{i3} +\eta\xi_i\bm{v}_{i4},
    \end{align}
    where
    \[
    \begin{split}
    \bm{v}_{i1} &= \sum_{j=1}^m \sum_{k=1}^s\bm{b}_j\bm{b}_j^*\widetilde{\bm{h}}_k^t\left(\widetilde{\bm{x}}_k^{t*}\bm{a}_{kj}\bm{a}_{ij}^*\widetilde{\bm{x}}_i^t - {\bm{x}}_k^{\natural*}\bm{a}_{kj}\bm{a}_{ij}^*{\bm{x}}_i^\natural \right)\\
    \bm{v}_{i2}& =\sum_{j=1}^m\sum_{k=1}^s\bm{b}_j\bm{b}_j^*\widetilde{\bm{h}}_k^t\left({\bm{x}}_k^{\natural*}\bm{a}_{kj}\bm{a}_{ij}^*{\bm{x}}_i^\natural-\|\bm{x}_k^\natural\|_2^2\right)\\
    \bm{v}_{i3}& = \sum_{j=1}^m \sum_{k=1}^s\bm{b}_j\bm{b}_j^*\bm{h}_k^\natural\bm{x}_k^{\natural*}\bm{a}_{kj}\bm{a}_{ij}^*\widetilde{\bm{x}}_i^t\\
    \bm{v}_{i4}& = \sum_{j=1}^m e_j\bm{b}_j\bm{a}_{ij}^*\widetilde{\bm{x}}^t_i,
    \end{split}
    \]which is based on the fact that $\sum_{j=1}^m \bm{b}_j\bm{b}_j^* = \bm{I}_K$.
    In what follows, we bound the above four terms respectively.
    \begin{enumerate}[1.]
        \item  
        Based on the inductive hypothesis (\ref{h}), the incoherence inequality (\ref{inco}) and the concentration inequality  \cite{ma2017implicit}
        \begin{align}\label{121}
        \max_{1\leq i \leq s, 1\leq j\leq m}\left|\bm{a}_{ij}^*\bm{x}_i^\natural\right|\cdot\|\bm{x}_i^\natural\|_2^{-1}\leq 5 \sqrt{\log m},
        \end{align}
        with the probability at least $1-O(m^{-10})$, we have
        \begin{align}
        |\bm{b}_l^*\bm{v}_{i1}|\cdot\|\bm{h}_i^\natural\|_2^{-1}\leq0.1s\max_{1\leq k\leq s, 1\leq j \leq m}|\bm{b}_j^*\widetilde{\bm{h}}_k^t|\cdot\|\bm{h}_i^\natural\|_2^{-1},
        \end{align}
        as long as $C_3$ is sufficiently small,
        \begin{align}
    |\bm{b}_1^*\bm{v}_{i2}|\cdot\|\bm{h}_i^\natural\|_2^{-1}\leq&(0.1+0.1\sqrt{s})\max_{1 \leq k\leq s, 1\leq l\leq m}\left|\bm{b}_l^*\widetilde{ \bm{h}}_k^t\right|\cdot \notag\\
    &\|\bm{h}_i^\natural\|_2^{-1}+O(cC_4\frac{s\mu}{\sqrt{m}}\log^2 m),
    \end{align}
    as long as $m\gg s^2K\log^2 m$ with some sufficiently large constant $C_4>0$ and some sufficiently small constant $c>0$,
    \begin{align}
    |\bm{b}_l^*\bm{v}_{i3}|\cdot\|\bm{h}_i^\natural\|_2^{-1}
    \lesssim(1+C_3\sqrt{s})\frac{\mu}{\sqrt{m}},
    \end{align}
    as long as picking up sufficiently small $C_3>0$.
        \item We end the proof with controlling $|\bm{b}_l^*\bm{v}_{i4}|$:
        \begin{align}
        |\bm{b}_l^*\bm{v}_{i4}|\cdot\|\bm{h}_i^\natural\|_2^{-1}
        & \leq \sum_{j=1}^m |\bm{b}_l^*\bm{b}_j|\left\{\max_{1 \leq k\leq s, 1\leq j\leq m}\frac{|\bm{a}_{kj}^*\widetilde{\bm{x}}_k^t|}{\|\bm{x}_i^\natural\|_2}\right\}|e_j|\notag\\
        & \overset{(\text{i})}{\lesssim} \sigma^2\frac{\log^{3/2}m}{m}\leq \log m,
        \end{align}
        as long as $m\gg \sigma^2\sqrt{\log m}$. Here the step (i)  arises from the inequality that  with probability at least $1-O(m^{-10})$,
        \begin{align}\label{ax}
        &\max_{1\leq k\leq s, 1\leq j \leq m}\left|\bm{a}_{kj}^*\widetilde{\bm{x}}_k^t\right|\cdot\|\bm{x}_i^\natural\|_2^{-1}\notag\\
        &\leq\max_{1\leq k\leq s, 1\leq j \leq m}\frac{\left|\bm{a}_{kj}^*(\widetilde{\bm{x}}_k^t-\bm{x}_k^\natural)\right|}{\|\bm{x}_k^\natural\|_2}+\max_{1\leq k\leq s, 1\leq j \leq m}\frac{\left|\bm{a}_{kj}^*{\bm{x}}_k^\natural\right|}{\|\bm{x}_k^\natural\|_2}\notag\\
        &\leq 6\sqrt{\log m},
        \end{align}
        as long as $m$ is sufficiently large, the inequality that $
        \sum_{j=1}^m |\bm{b}_l^*\bm{b}_j|\leq 4\log m
        $ \cite[Lemma 48]{ma2017implicit}, and the assumption $|e_j| \leq \sigma^2/m\ll 1$ provided in Section \ref{form}.
    \end{enumerate}
     Putting the above results together, there exists some constant $C_8>0$ such that 
        \begin{align}
        \frac{\left|\bm{b}_l^*\widetilde{\bm{h}}_i^{t+1}\right|}{\|\bm{h}_i^\natural\|_2}&\leq(1+\delta) \Bigg\{\bigg(|\bm{b}_l^*\widetilde{ \bm{h}}_i^t|-\eta\xi_i\sum_{k=1}^s{|\bm{b}_l^*\widetilde{ \bm{h}}_k^t|}+(1+0.1\sqrt{s}\notag\\
        +0.1s)\!\!&\max_{1\leq k\leq s, 1\leq j \leq m}|\bm{b}_j^*\widetilde{\bm{h}}_k^t|\bigg)\cdot\|\bm{h}_i^\natural\|_2^{-1}+C_8(1+C_3\sqrt{s})\cdot\notag\\
        &\eta\xi_i\frac{\mu}{\sqrt{m}}+C_8cC_4\eta\xi_i\frac{s\mu}{\sqrt{m}}\log^2 m+C_8\eta\xi_i\log m\Bigg\} \notag\\
        &\leq C_4\frac{\mu}{\sqrt{m}}\log^2m.
        \end{align}
        The last step holds as long as $c>0$ is sufficiently small, i.e., $(1+\delta)C_8\eta\xi_ic\gg 1$, and the stepsize obeys $\eta>0$ and $\eta\asymp s^{-1}.$ To accomplish  the proof, we need to pick the sample size that
        $
        m\gg( \mu^2+\sigma^2)\tau K\log^4m,
    $
        where $\tau = c_{10}s^2\log^4m$ with some sufficiently large constant $c_{10}>0$.

    \section{Proof of Lemma \ref{L6}}\label{proof6}
    Recall that $\check{\bm{h}}_i^0$ and $\check{\bm{x}}_i^0$ are the leading left and right singular vectors of $\bm{M}_i$, $i=1,\cdots,s$, where
   $
    \bm{M}_i =  \sum_{j=1}^m\sum_{k=1}^s\bm{b}_j\bm{b}_j^*\bm{h}_k^\natural\bm{x}_k^{\natural*}\bm{a}_{kj}\bm{a}_{ij}^*+\sum_{j=1}^m e_j\bm{b}_j\bm{a}_{ij}^*
   $. By exploiting a variant of Wedin's sin$\Theta$ theorem \cite[Therorem 2.1]{dopico2000note}, we derive that
    \begin{align}\label{149}
    &\min_{\alpha_i\in\mathbb{C},|\alpha_i| =1}\left\| \alpha_i\check{\bm{h}}_i^0-\bm{h}_i^{\natural}\right\|_2+ \left\|\alpha_i\check{\bm{x}}_i^0-\bm{x}_i^{\natural}\right\|_2\notag\\
    \leq& \frac{c_1\|\bm{M}_i - \mathbb{ E}[\bm{M}_i]\|}{\sigma_1(\mathbb{ E}[\bm{M}_i]) - \sigma_2(\bm{M}_i)},
    \end{align}
    for some constant $c_1>0$, where $\sigma_1(\bm{A})$ and $\sigma_2(\bm{A})$ denote the largest eigenvalue and second largest eigenvalue of the matrix $\bm{A}$. In the view of the numerator of (\ref{149}), it has been specified in \cite[Lemma 6.16]{ling2018regularized} that for any $\xi>0$, 
    \begin{align}\label{150}
    \|\bm{M}_i-\mathbb{ E}[\bm{M}_i]\|\leq \frac{\xi}{\kappa\sqrt{s}},
    \end{align}
    with probability at least $1-O(m^{-10})$, provided that 
    $
    m \gg {c_2(\mu^2+\sigma^2)s\kappa^2K\log m}/{\xi^2},
  $
    for some constant $c_2>0$. Inspired by the technical method used in \cite[Section C.5]{ma2017implicit}. We further bound the denominator of (\ref{149}) via combining (\ref{150}) and Weyl's inequality, derived as
    $
    \sigma_1(\mathbb{ E}[\bm{M}_i]) - \sigma_2(\bm{M}_i)
    \geq 1-\frac{\xi}{\kappa\sqrt{s}}.
    $
     We then get  
    \begin{align}\label{153}
    \min_{\alpha_i\in\mathbb{C},|\alpha_i| =1}\left\| \alpha_i\check{\bm{h}}_i^0-\bm{h}_i^{\natural}\right\|_2+ \left\|\alpha_i\check{\bm{x}}_i^0-\bm{x}_i^{\natural}\right\|_2 \leq 2c_1\frac{\xi}{\kappa\sqrt{s}},
    \end{align}
    as long as $\xi < 1/2$. Moreover, we extend the bound (\ref{153}) to the inequality with the scaled singular vector $\bm{h}_i^0 = \sqrt{\sigma_1(\bm{M}_i)}\check{\bm{h}}_k^0$ and $\bm{x}_i^0 = \sqrt{\sigma_1(\bm{M}_i)}\check{\bm{x}}_k^0$ via using the inequality provided in \cite[Section C.5]{ma2017implicit}. It yields that
    \begin{align}
    &\left\| \alpha_i{\bm{h}}_i^0-\bm{h}_i^{\natural}\right\|_2+ \left\|\alpha_i{\bm{x}}_i^0-\bm{x}_i^{\natural}\right\|_2\notag\\
    \leq
    & \left\| \alpha_i\check{\bm{h}}_i^0-\bm{h}_i^{\natural}\right\|_2+ \left\|\alpha_i\check{\bm{x}}_i^0-\bm{x}_i^{\natural}\right\|_2 + 2\frac{\xi}{\kappa\sqrt{s}}.
    \end{align}
    We thus conclude that 
    \begin{align}
    &\min_{\alpha_i\in\mathbb{C}|\alpha_i| = 1}\left\{\left\|\alpha_i\bm{h}_i^0-\bm{h}_i^\natural\right\|_2+\left\|\alpha_i\bm{x}_i^0-\bm{x}_i^\natural\right\|_2\right\}\notag\\
    \leq& 2c_1\frac{\xi}{\kappa\sqrt{s}}+2\frac{\xi}{\kappa\sqrt{s}}.
    \end{align}
    Since $\xi$ is arbitrary, we accomplish the proof for (\ref{26}) by taking $m\gg (\mu^2+\sigma^2)s\kappa^2K\log m$. Under similar arguments, we can also establish (\ref{27}) in Lemma \ref{L6}, which is omitted here. We further obtain the last claim in Lemma \ref{L6} via combining the inequality (\ref{26}) and \cite[Lemma 54]{ma2017implicit}, given as
$
    ||\alpha_i^0|-1|\lesssim\frac{\xi}{\kappa\sqrt{s}}< 1/4,~ 1\leq i\leq s.
$
    \section{Proof of Lemma \ref{L7}}\label{proof7}
    With the similar strategy in \cite[Section C.6]{ma2017implicit}, we first show that the normalized singular vectors of $\bm{M}_i$ and $\bm{M}_i^{(l)}$, $i = 1,\cdots, s$ are close enough. We further extend this inequality to the scaled singular vectors, thereby converting the $\ell_2$ metric to the distance function defined in (\ref{dist}). We finally prove the incoherence of $\{\bm{h}_i\}_{i=1}^s$ with respect to $\{\bm{b}_j\}_{j=1}^m$.
    
    Recall that $\check{\bm{h}}_i^0$ and $\check{\bm{x}}_i^0$ are the leading left and right singular vectors of $\bm{M}_i$, $i=1,\cdots,s$, and $\check{\bm{h}}_i^{0,(l)}$ and $\check{\bm{x}}_i^{0,(l)}$ are the leading left and right singular vectors of $\bm{M}_i^{(l)}$, $i=1,\cdots,s$. By exploiting a variant of Wedin's sin$\Theta$ theorem \cite[Therorem 2.1]{dopico2000note}, we derive that
    \begin{align}\label{158}
    &\min_{\alpha_i\in\mathbb{C},|\alpha_i| =1}\left\| \alpha_i\check{\bm{h}}_i^0-\check{\bm{h}}_i^{0,(l)}\right\|_2+ \left\|\alpha_i\check{\bm{x}}_i^0-\check{\bm{x}}_i^{0,(l)}\right\|_2\notag\\
    \leq& \frac{c_1\left\|(\bm{M}_i-  \bm{M}_i^{(l)})\check{\bm{x}}_i^{0,(l)}\right\|_2 + c_1\left\|\check{\bm{h}}_i^{0,(l)*}(\bm{M}_i-  \bm{M}_i^{(l)})\right\|_2}{\sigma_1(\bm{M}_i^{(l)}) - \sigma_2(\bm{M}_i)},
    \end{align}
    for $i=1,\cdots,s$ with some constant $c_1>0$. According to \cite[Section C.6]{ma2017implicit}, for $i = 1,\cdots, s$, we have
    \begin{align}
    &\sigma_1(\bm{M}_i^{(l)}) - \sigma_2(\bm{M}_i)\notag\\
    \geq&~ 3/4 - \|\bm{M}_i^{(l)} - \mathbb{ E}[\bm{M}_i^{(l)}]\| - \|\bm{M}_i - \mathbb{ E}[\bm{M}_i]\|\geq 1/2,
    \end{align}
    where the last step comes from \cite[Lemma 6.16]{ling2018regularized} provided that $m\gg (\mu^2+\sigma^2)sK\log m$. As a result, we obtain that for $i=1,\cdots,s$
    \begin{align}\label{161}
    &\left\| \beta_i^{0,(l)}\check{\bm{h}}_i^0-\check{\bm{h}}_i^{0,(l)}\right\|_2+ \left\|\beta_i^{0,(l)}\check{\bm{x}}_i^0-\check{\bm{x}}_i^{0,(l)}\right\|_2\notag\\
    \leq&2c_1\left\{\left\|(\bm{M}_i-  \bm{M}_i^{(l)})\check{\bm{x}}_i^{0,(l)}\right\|_2 + \left\|\check{\bm{h}}_i^{0,(l)*}(\bm{M}_i-  \bm{M}_i^{(l)})\right\|_2\right\},
    \end{align}
    where 
    \begin{align}\label{162}
    \beta_i^{0,(l)}:=\argmin_{\alpha\in\mathbb{C},|\alpha| =1}\left\| \alpha\check{\bm{h}}_i^0-\check{\bm{h}}_i^{0,(l)}\right\|_2+ \left\|\alpha\check{\bm{x}}_i^0-\check{\bm{x}}_i^{0,(l)}\right\|_2.
    \end{align}
    It thus suffices to control the two terms on the right-hand side of (\ref{161}). Therein, 
    \begin{align}
    \bm{M}_i-\bm{M}_i^{(l)} = \bm{b}_l\bm{b}_l^*\sum_{k=1}^s\bm{h}_k^{\natural}\bm{x}_k^{\natural*}\bm{a}_{kl}\bm{a}_{il}^* + e_l\bm{b}_l\bm{a}_{il}^*.
    \end{align}
    Inspired the similar strategy used in \cite[Section C.6]{ma2017implicit}, we conclude that
        \begin{align}\label{1661}
    &\left\| \beta_i^{0,(l)}\check{\bm{h}}_i^0-\check{\bm{h}}_i^{0,(l)}\right\|_2+ \left\|\beta_i^{0,(l)}\check{\bm{x}}_i^0-\check{\bm{x}}_i^{0,(l)}\right\|_2\notag\\
    \leq & 2C_1\Bigg\{30\frac{\mu}{\sqrt{m}}\cdot \sqrt{\frac{s^2K\log^2 m}{m}}+ \frac{5\sigma^2}{m}\sqrt{\frac{K\log m}{m}}\notag\\
    &\left( 15\sqrt{\frac{\mu^2s^2K\log m}{m}}+3\sqrt{K}\frac{\sigma^2}{m}\right)|\bm{b}_l^*\check{\bm{h}}_i^0|\cdot\|\bm{h}_i^\natural\|_2^{-1}+\notag\\&\left(15\sqrt{\frac{\mu^2s^2K\log m}{m}}
    \sqrt{\frac{K}{m}}+3\sqrt{K}\frac{\sigma^2}{m}\right)\kappa\left\|\widetilde{\alpha}_i\check{\bm{h}}_i^0-\check{\bm{h}}_i^{0,(l)}\right\|_2
    \Bigg\}
    \end{align}
     via exploiting the fact that $\|\bm{b}_l\|_2 = \sqrt{K/m}$, the incoherence condition (\ref{inco}), the bound (\ref{121}), the assumption $|e_j|\leq \frac{\sigma^2}{m}\ll 1$ provided in Section \ref{form} and the condition that with probability exceeding $1-O(m^{-10})$,
    \begin{align}\label{165}
    \max_{1\leq l \leq m}|\bm{a}_{il}^*\check{\bm{x}}_i^{0,(l)}|\cdot\|\bm{x}_i^\natural\|_2^{-1}\leq 5\sqrt{\log m},
    \end{align}
    due to the independence between $\check{\bm{x}}_i^{0,(l)}$ and $\bm{a}_{il}$ \cite[Section C.6]{ma2017implicit}.
    
        Since the inequality (\ref{1661}) holds for any $|\widetilde{ \alpha}_i|=1$, we can pick up $\widetilde{ \alpha}_i = \beta^{0,(l)}$.
        With the assumption that $m\gg (\mu+\sigma^2) s\kappa K\log^{1/2} m$ such that $1-30c_1\kappa\sqrt{\frac{\mu^2s^2K\log m}{m}}\cdot
        \sqrt{\frac{K}{m}}-6\kappa\sqrt{K}\frac{\sigma^2}{m}\leq \frac{1}{2}$, we get 
        \begin{align}\label{169}
        &\max_{1 \leq i\leq s, 1\leq j\leq m}\left\| \beta_i^{0,(l)}\check{\bm{h}}_i^0-\check{\bm{h}}_i^{0,(l)}\right\|_2+ \left\|\beta_i^{0,(l)}\check{\bm{x}}_i^0-\check{\bm{x}}_i^{0,(l)}\right\|_2\notag\\
        \leq &120c_1\frac{\mu}{\sqrt{m}}\cdot \sqrt{\frac{s^2K\log^2 m}{m}}+ \frac{20c_1\sigma^2}{m}\sqrt{\frac{K\log m}{m}}\notag\\
        &+\left( 60c_1\sqrt{\frac{\mu^2s^2K\log m}{m}}+12c_1\sqrt{K}\frac{\sigma^2}{m}\right)\cdot\notag\\
        &~~\max_{1 \leq i\leq s, 1\leq j\leq m}|\bm{b}_l^*\check{\bm{h}}_i^0|\cdot\|\bm{h}_i^\natural\|_2^{-1}.
        \end{align}
        It thus suffices to control $\max_{1 \leq i\leq s, 1\leq j\leq m}|\bm{b}_l^*\check{\bm{h}}_i^0|\cdot\|\bm{h}_i^\natural\|_2^{-1}$. We further define that $\bm{M}_i\check{\bm{x}}^0 = \sigma_1(\bm{M}_i)\check{\bm{h}}_i^0$ and 
    $
        \bm{W}_i = \sum_{j=1}^m\bm{b}_j(\sum_{k\neq i}\bm{b}_j^*\bm{h}_k^\natural\bm{x}_k^{\natural*}\bm{a}_{kj}+e_j)\bm{a}_{ij}^*,
A    $
        which further leads to 
        \begin{align}\label{234}
        &\max_{1 \leq i\leq s}|\bm{b}_l^*\check{\bm{h}}_i^0|\cdot\|\bm{h}_i^\natural\|_2^{-1}\notag\\
        =& \frac{1}{\sigma_1(\bm{M}_i)\cdot\|\bm{h}_i^\natural\|_2}|\bm{b}_l^*\bm{M}_i\check{\bm{x}}_i^0|\notag\\
        \overset{(\text{i})}{\leq} &2\left(\sum_{j=1}^m|\bm{b}_l^*\bm{b}_j|\right)\max_{ 1\leq i\leq s, 1\leq j\leq m}\left\{|\bm{b}_j^*\bm{h}_i^\natural|\cdot|\bm{a}_{ij}^*\bm{x}_i^\natural|\cdot|\bm{a}_{ij}^*\check{\bm{x}}_i^0|\right\}\cdot\notag\\&\|\bm{h}_i^\natural\|_2^{-1}+2\|\bm{b}_l\|_2\cdot\|\bm{W}_i\|\cdot\|\check{\bm{x}}_i^0\|_2\cdot\|\bm{h}_i^\natural\|_2^{-1}\cdot\notag\\
        &\max_{ 1\leq i\leq s, 1\leq j\leq m}\bigg\{\left|\bm{a}_j^*\check{\bm{x}}_i^{0,(j)}\right|+\|\bm{a}_{ij}\|_2\left\|\beta_i^{0,(j)}\check{\bm{x}}_i^0-\check{\bm{x}}_i^{0,(j)}\right\|_2\bigg\}\notag\\
        \overset{(\text{ii})}{\leq}& \kappa\sqrt{\frac{K}{m\log m}}+200\frac{\mu\log^2m}{\sqrt{m}} +120\kappa\sqrt{\frac{\mu^2K\log^3m}{m}}\cdot\notag\\&\max_{ 1\leq i\leq s, 1\leq j\leq m}\left\|\beta_i^{0,(j)}\check{\bm{x}}_i^0-\check{\bm{x}}_i^{0,(j)}\right\|_2,
        \end{align}
        where $\beta_i^{0,(j)}$ is defined in (\ref{162}). Here, (i) arises from the low bound $\sigma_1(\bm{M}_i)\geq \frac{1}{2}$, the triangle inequality and the Cauchy-Schwarz inequality. The step (ii) comes from combining the assumption that $\|\bm{h}_i^\natural\|_2=\|\bm{x}_i^\natural\|_2$, for $i=1,\cdots,s $, $\max_{1\leq i \leq s}\|\bm{h}_i^\natural\|_2=1$, the incoherence condition (\ref{inco}), the bound (\ref{121}), the triangle inequality, the estimate: $\sum_{j=1}^m |\bm{b}_l^*\bm{b}_j|\leq 4\log m$ \cite[Lemma 48]{ma2017implicit}, $\|\bm{b}_l\| = \sqrt{{K}/{m}}$, $\|\check{\bm{x}}_i^0\|_2 = 1$, the inequality (\ref{165}) and the bound that with probability $1-O(m^{-9})$ \cite{ling2018regularized}, 
        \begin{align}\label{184}
        \|\bm{W}_i\|\leq \frac{\|\bm{h}_i^\natural\|_2\cdot\|\bm{x}_i^\natural\|_2}{2\sqrt{\log m}},
        \end{align}
        if $m\gg (\mu^2+\sigma^2)sK\log^2m$.
        Combining the bound (\ref{169}) and (\ref{234}) and the assumption $m\gg (\mu^2+\sigma^2)s^2\kappa K\log^2m$ such that
        $
        ( 60c_1\sqrt{\frac{\mu^2s^2K\log m}{m}}+12c_1\sqrt{K}\frac{\sigma^2}{m})\cdot120\kappa\sqrt{\frac{\mu^2K\log^3m}{m}}\leq 1/2,
    $
    we have
        \begin{align}\label{175}
        &\max_{1 \leq i\leq s,1\leq l\leq m}\left\| \beta_i^{0,(l)}\check{\bm{h}}_i^0-\check{\bm{h}}_i^{0,(l)}\right\|_2+ \left\|\beta_i^{0,(l)}\check{\bm{x}}_i^0-\check{\bm{x}}_i^{0,(l)}\right\|_2\notag\\
        \leq & C_4\frac{\mu}{\sqrt{m}}\sqrt{\frac{\mu^2s^2K\log^5m}{m}},
        \end{align}
        for some constant $C_4>0$. Taking the bound (\ref{175}) together with (\ref{234}), it yields 
    $
        \max_{1 \leq i\leq s, 1\leq l\leq m}|\bm{b}_l^*\check{\bm{h}}_i^0|\|\bm{h}_i^\natural\|_2^{-1}
        \leq~c_2 \frac{\mu\log^2m}{\sqrt{m}},
  $
        for some constant $c_2>0$, as long as $m\gg (\mu^2+\sigma^2)s\kappa^2K \log^2m$.
         
        We further scaled the preceding bounds to the final version. Based on \cite[Section C.6]{ma2017implicit}, one has
        \begin{align}\label{scale}
        &\left\|\alpha\bm{h}^0 - \bm{h}^{0,(l)}\right\|_2 + \left\|\alpha\bm{x}^0 - \bm{x}^{0,(l)}\right\|_2\notag\\
        \leq&\left\|(\bm{M}_i-  \bm{M}_i^{(l)})\check{\bm{x}}_i^{0,(l)}\right\|_2+
        6\bigg\{\left\| \alpha\check{\bm{h}}_i^0-\check{\bm{h}}_i^{0,(l)}\right\|_2\notag\\&+ \left\|\alpha\check{\bm{x}}_i^0-\check{\bm{x}}_i^{0,(l)}\right\|_2\bigg\}.
        \end{align}
    Taking the bounds (\ref{175}) and (\ref{scale}) collectively yields
    \begin{align}
    &\min_{\alpha_i\in\mathbb{C},|\alpha_i| =1}\left\| \alpha_i{\bm{h}}_i^0-\bm{h}_i^{0,(l)}\right\|_2+ \left\|\alpha_i{\bm{x}}_i^0-\bm{x}_i^{0,(l)}\right\|_2\notag\\
    \leq&~c_5\frac{\mu}{\sqrt{m}}\sqrt{\frac{\mu^2s^2K\log^5m}{m}},
    \end{align}
    for some constant $c_5>0$, as long as $m\gg (\mu^2+\sigma^2)s^2K\log^2m$.
    
    Furthermore, by exploiting the technical methods provided in \cite[Section C.6]{ma2017implicit}, we have
   $
    \mathrm{dist}\left(\bm{z}^{0,(l)},\widetilde{\bm{z}}^0\right)
    \leq 4c_5\frac{s\kappa\mu}{\sqrt{m}}\sqrt{\frac{\mu^2sK\log^5m}{m}}.
 $
         This accomplishes the proof for the claim (\ref{28}). We further move to the proof for the claim (\ref{29}). In terms of $|\bm{b}_l^*\widetilde{ \bm{h}}_i^0|$, one has
    \begin{align}
    \frac{|\bm{b}_l^*\widetilde{ \bm{h}}_i^0|}{\|\bm{h}_i^\natural\|_2}&\leq \frac{\left|\bm{b}_l^*\frac{1}{\overline{\alpha_i^0}}\bm{h}_i^0\right|}{\|\bm{h}_i^\natural\|_2}\leq\left|\frac{1}{\overline{\alpha_i^0}}\right|\frac{|\bm{b}_l^*{ \bm{h}}_i^0|}{\|\bm{h}_i^\natural\|_2}\leq 2 \frac{\left|\sqrt{\sigma_1(\bm{M}_i)}\bm{b}_l^*\check{\bm{h}}_i^0\right|}{\|\bm{h}_i^\natural\|_2}\notag\\
    &\leq2\sqrt{2}c_2 \frac{\mu\log^2m}{\sqrt{m}},
    \end{align}
    based on fact that 
    $
    \frac{1}{2}\leq  \sigma_1(\bm{M}_i) \leq 2.
  $

 \bibliographystyle{ieeetr}

\end{document}